\documentclass[aps,twocolumn,showpacs,preprintnumbers,prb]{revtex4-1}
\usepackage{dcolumn}
\usepackage{graphicx}
\usepackage{amsmath, amssymb}
\usepackage{bm}
\usepackage{mathrsfs}
\usepackage{subfigure}
\usepackage{dsfont}
\usepackage{latexsym}
\usepackage{amsthm}
\usepackage[colorlinks,linkcolor=blue,anchorcolor=blue,citecolor=blue]{hyperref}
\usepackage{lipsum}

\newcommand{\calH}{{\mathcal H}}

\newcommand{\calK}{{\mathcal K}}
\newcommand{\calL}{{\mathcal L}}
\newcommand{\calU}{{\mathcal U}}
\newcommand{\kappamax}{\kappa_{\mbox{\scriptsize max}}}
\newtheorem{theorem}{Theorem}

\begin{document}

\title{Particle Statistics, Frustration, and Ground-State Energy}

\author{Wenxing Nie}
\affiliation{Center for Theoretical Physics, Department of Physics, Sichuan University, Chengdu 610064, China}
\author{Hosho Katsura}
\affiliation{Department of Physics, University of Tokyo, Hongo, Bunkyo-ku, Tokyo 113-0033, Japan}
\author{Masaki Oshikawa}
\affiliation{Institute for Solid State Physics, University of Tokyo,
Kashiwa 277-8581, Japan}

\date{\today}

\begin{abstract}

We study the connections among particle statistics, frustration, and
ground-state energy in quantum many-particle systems.
In the absence of interaction,
the influence of particle statistics on the ground-state energy is
trivial: the ground-state energy of noninteracting bosons is lower than
that of free fermions because of Bose-Einstein condensation (BEC) and Pauli exclusion principle. In the presence of hard-core or other interaction,
however, the comparison
is not trivial. Nevertheless, the ground-state energy of hard-core
bosons is proved to be lower than that of spinless fermions, if
all the hopping amplitudes are nonnegative.
The condition can be understood as the absence of frustration
among hoppings. By mapping the many-body
Hamiltonian to a tight-binding model on a fictitious lattice,
we show that the Fermi statistics of the original particles introduces
an effective magnetic flux in the fictitious lattice.
The latter can be effectively regarded as a kind of frustration,
since it leads to a destructive interference among different paths along which a single particle is propagating.
If we introduce hopping frustration, the hopping frustration is expected to
compete with the effective frustration due to the Fermi statistics,
leading to the possibility that the ground-state energy of
hard-core bosons can be higher than that of
fermions. We present several examples, in which
the ground-state energy of hard-core bosons is proved to be higher
than that of fermions due to the hopping frustration.
The basic ideas were reported in the preceding Letter [W.-X. Nie, H. Katsura, and M. Oshikawa,
Phys. Rev. Lett. {\bf 111}, 100402 (2013)]; more details and
several extensions, including one to the spinful case, are discussed
in the present paper.

\end{abstract}
\date{\today}
\maketitle

\section{\label{intro}Introduction}

In this paper, we study a simple question:
how the particle statistics affects the ground-state energy of the
system.
More specifically, we compare the ground-state energy of
bosons and fermions on an identical lattice with same
parameters such as hopping amplitudes.

In noninteracting systems, the influence of particle statistics on the
ground-state energy can be understood easily.  The two systems in
comparison are exactly equivalent to each other, and thus have exactly
the same ground-state energy, when only a single particle is present.
The ground-state energy of fermions is simply given by the sum of the
lowest single-particle energy eigenvalues, following the Aufbau
principle.  In contrast, in the ground state of noninteracting bosons,
all the bosons condense into the lowest single-particle state. This
phenomenon is known as
Bose-Einstein Condensation (BEC). Therefore, the
ground state energies of
non-interacting bosons and fermions satisfy the ``natural''
inequality:
\begin{equation}
E_0^{\rm B}\leq E_0^{\rm F}.  \label{eq.e0b.leq.e0f}
\end{equation}

On the other hand, the comparison of the ground-state energies of bosons
and fermions is not trivial in the presence of interaction, because the
simple argument based on the perfect BEC breaks down.
In a system of interacting bosons, it is in fact already a nontrivial question whether
the BEC actually takes place.  Einstein's original argument depends on
the absence of interaction.  For interacting bosons, there is no general
theorem that BEC always occurs~\cite{Leggett}.  A counterexample is the
solid ${}^4$He phase, where BEC is absent even at zero temperature,
under a sufficiently high pressure.  Rigorously proven examples of BEC,
in the sense of the off-diagonal long-range order (ODLRO),
in  interacting systems are still rather
limited~\cite{Kennedy1988,LiebSeiringer2002,Aizenman2004}.  Even if the
occurrence of BEC or the ODLRO is proved in a
system of interacting bosons, it does not necessarily restrict
the ground-state energy, because single-particle states with higher energies
can be partially occupied. In particular, an ODLRO does not necessarily
imply the inequality~\eqref{eq.e0b.leq.e0f}.
In fact, the influence of particle statistics on the ground-state energy had not been
much explored in strongly correlated systems.

The comparison of the ground-state energies is particularly appealing
in the case of hard-core bosons and fermions. In both kinds of systems,
each site is either empty or occupied by a single particle.
Thus the dimension of the Hilbert space is identical between
them. Nevertheless, the different particle statistics generically
lead to different ground-state energies, as we will see in the following.

Before discussing the issue any further, let us comment on the physical
relevance of the question itself.
The energy eigenvalue itself is generally unphysical in the sense that
one can always redefine the energy by adding a constant.
It is thus the difference of energies of two different states
that matters.

We can understand the difference by defining the ground-state
energy with respect to a simple reference state in each system,
such as a vacuum state (in which every site is empty).
This ground-state energy is the sum of energy gains in the process
of filling the system with particles,
and is a measurable quantity~\cite{He_expe1,He_expe2}.
This is somewhat similar to the
``enthalpy of formation'' studied in chemistry~\cite{Atkins},
which is the total change  of enthalpy (per mole) when the compound
is formed from its elements under a certain condition.

Since the vacuum state is equivalent between the system of bosons
and fermions, the comparison of the ground-state energies is completely
well defined.
Moreover, even when the energy difference itself cannot be measured,
the comparison of the ground-state energies is relevant for
understanding stability of various different phases.
This is particularly the case with the possible realization of
statistical transmutation, as we will discuss later in this paper.

Concerning the comparison of the ground-state energies between
bosons and fermions, recently we found~\cite{ourPRL} a sufficient condition for
the natural inequality~\eqref{eq.e0b.leq.e0f} to hold,
without relying on the occurrence of BEC.
That is, if all the hopping amplitudes are nonnegative,
the ground-state energy of hard-core bosons
is still lower than that of the corresponding fermions.
This theorem is extended to the spinful case
in the present paper.
Once we relax the condition of nonnegative hopping amplitudes,
it is possible to reverse the inequality so that
the ground-state energy of bosons
is higher than that of fermions.
We find several concrete models in which such a reversal
is realized; and in several cases it is even proved rigorously.
More examples and techniques will be
introduced in the present paper, than those discussed in
Ref.~\onlinecite{ourPRL}.

Moreover, our study leads to a novel physical understanding of the
effects of particle statistics, in terms of \emph{frustration} in
quantal phase.  This is more general than the picture based on the
perfect BEC, and is indeed applicable to systems with interaction.

We can map a quantum many-particle problem to a
single-particle problem on a fictitious lattice in higher dimensions.
When all the hopping amplitudes are nonnegative and the particles
are bosons, the corresponding single-particle problem also
has only nonnegative hopping amplitudes.
In such a case, there is no frustration in the
quantal phase of the wavefunction.
On the other hand,
Fermi statistics of the original particles gives an effective
magnetic flux in the corresponding single-particle problem.
This implies a frustration in the phase of the wavefunction,
induced by the Fermi statistics.
When a magnetic flux is introduced in the original quantum
many-particle problem, it also results in a magnetic flux
in the corresponding single-particle problem, inducing
a frustration.
This hopping-induced frustration and the
the effective frustration induced by the Fermi statistics
can sometimes partially cancel with each
other, resulting in the reversed inequality between
the ground-state energies of the hard-core
bosons and fermions.

The paper is organized as follows. In Sec.~\ref{natural}, we present the
full proof of the natural inequality for the spinless case and extend the
discussion to the spinful case. Based on the proof, in
Sec.~\ref{frustration}, we put forward a unified understanding of
the frustration for bosons and fermions in the same manner.
As a by-product, a strict version of the diamagnetic inequality
for a general lattice is presented.
Several examples, in which the natural inequality
is violated owing to the hopping frustration,
are presented in Sec.~\ref{reversed}.
The examples include a simple yet instructive, exactly
solvable model of particles on a one-dimensional ring,
two-dimensional systems of coupled rings, 
systems with flux in $2$D and $3$D, 
and flat band models.
Rigorous proof of the reversed inequality is provided for most cases.
Conclusions and discussions are
presented in Sec.~\ref{conclusions}.
Detailed proofs of some of the theorems, and related technical
results are presented in Appendices.

\section{\label{natural}Natural inequality}
The natural inequality~\eqref{eq.e0b.leq.e0f} holds
trivially for noninteracting bosons and fermions
with the same form of the Hamiltonian.
Now we present three theorems, which state that the
Eq.~\eqref{eq.e0b.leq.e0f} holds even for hard-core bosons,
provided that all the hopping amplitudes are nonnegative.
A brief overview appeared in Ref.~\onlinecite{ourPRL}, but
here we give a more detailed discussion, and also
an extension to the spinful case.

\subsection{\label{spinless}Natural inequality for spinless case}
First we consider the comparison of spinless hard-core bosons
with spinless fermions.
We assume the system of bosons or fermions is
described by the same form of Hamiltonian,
\begin{equation}
\calH=-\sum_{j\ne k}\left(
t_{jk}c^{\dagger}_jc_k+\mbox{H.c.}\right)
-\sum_j\mu_jn_j
+\sum_{j,k}V_{jk}n_jn_k,
\label{eq.Ham}
\end{equation}
where $j$ is the label of a site on a finite lattice $\Lambda$ and
$n_j\equiv c^{\dagger}_jc_j$ is the number of particles on $j$-th site.
Chemical potential $\mu$ is the uniform (site independent) part of $\mu_j$.
For a system of bosons, we identify $c_j$ with the boson
annihilation operator $b_j$ satisfying the standard commutation
relations, with the hard-core constraint $n_j = 0,1$ at
each site.  The hard-core constraint
may also be implemented by introducing
an infinite on-site interaction $\frac{\calU}{2}\sum_jn_j(n_j - 1)$, where $\calU\to+\infty$.
For a system of fermions, we identify $c_j$ with the fermion
annihilation operator $f_j$ satisfying the standard anticommutation
relations.

This Hamiltonian is very general. We do not make any assumption
on the dimensionality or the geometry of the lattice $\Lambda$, or
on the range of the hoppings.
In addition, the interaction is also arbitrary, as long
as it can be written in terms of $V_{jk}$.
The interesting aspect of attractive interaction will be discussed  
in Appendix~\ref{app:proof3}.
We note that the Hamiltonian~\eqref{eq.Ham}
conserves the total particle number.  Thus the ground state can be
defined for a given number of particles $M$ (canonical ensemble), or for
a given chemical potential $\mu$ (grand canonical ensemble).  The
comparison between bosons and fermions can be made in either
circumstance.

Now we will present a sufficient condition
for the natural inequality~\eqref{eq.e0b.leq.e0f}.
Moreover, a sufficient condition for the
strict inequality $E_0^{\rm B} < E_0^{\rm F}$ is provided.
The proof is also illuminating for physical understanding
of the natural inequality in interacting systems,
showing the importance of the particle statistics and
exchange processes.

\begin{theorem}\textup{(Natural inequality for spinless case)}

The inequality~\eqref{eq.e0b.leq.e0f} holds for any
given number of particles $M$ on a finite lattice $\Lambda$ with
$N \geq M$ sites,
if all the hopping amplitudes $t_{jk}$ are real and nonnegative.

Furthermore,
if the lattice $\Lambda$ is connected, and has a site directly connected
to three or more sites, and if the number of particles satisfies
$2 \leq M \leq N-2$, the strict inequality $E_0^{\rm B} < E_0^{\rm F}$ holds.
\label{thm.natural_spinless}
\end{theorem}

\begin{proof}
To write the matrix elements of the Hamiltonian~\eqref{eq.Ham}, we choose the occupation number basis
$|\phi^a\rangle\equiv|\{n^a_j\}\rangle$,
where $M$ is the total number of particles satisfying $\sum_j n^a_j = M$.
The matrix elements of the number operator $n_j$ are the same for hard-core bosons and spinless fermions in this basis.
We begin by defining the operator
\begin{equation}
\calK^{\rm B,F} \equiv - \calH^{\rm B,F} + C \mathds{1}.  \label{eq.K}
\end{equation}
For convenience, we added an identity matrix with large enough diagonal
elements $C$ such that all the eigenvalues $\kappa^{\rm B,F}$ of matrix
$\calK^{\rm B,F}$ and thus all the diagonal matrix elements
$\calK^{\rm B,F}_{aa}$ are positive.
The relation of the matrix elements for bosonic and fermionic operators
can be summarized as
\begin{equation}
\label{eq.Kab}
\calK_{ab}^{\rm B}=\left\{
\begin{array}{ll}
|\calK_{ab}^{\rm F}|&(a\ne b)\\
\calK_{aa}^{\rm F}& (a=b)
\end{array}
\right.=|\calK_{ab}^{\rm F}|.
\end{equation}
The difference between bosons and fermions is that,
given nonnegative hopping amplitudes $t_{jk}$,
the matrix elements of the bosonic operator
$\calK^{\rm B}$ is nonnegative,
while those of the fermionic operator $\calK^{\rm F}$
can be negative in sign.
This difference in signs generically leads to different
ground-state energies between bosons and fermions.

The ground state of the Hamiltonian $\calH^{\rm B,F}$ corresponds to the
eigenvector belonging to the largest eigenvalue $\kappamax^{\rm B,F}$
of $\calK^{\rm B,F}$.
Let $|\Psi_0\rangle_{\rm F}=\sum_a\psi_a|\phi^a\rangle_{\rm F}$
be the normalized ground state for fermions.
The trial state for the bosons can be assumed as
$|\Psi_0\rangle_{\rm B}=\sum_a|\psi_a| |\phi^a\rangle_{\rm B}$,
where $|\phi^a\rangle_{\rm B}$ is the basis state for bosons
corresponding to $|\phi^a\rangle_{\rm F}$.
Then, by a variational argument,
\begin{eqnarray}
\kappamax^{\rm B}&\ge&
{}_{\rm B}\langle\Psi_0|\calK^{\rm B}|\Psi_0\rangle_{\rm B}
=\sum_{ab}|\psi_a||\psi_b|\calK_{ab}^{\rm B}
\nonumber\\
&\ge& \sum_{ab} \psi_a^{*}\psi_b\calK_{ab}^{\rm F}=\kappamax^{\rm F}
\end{eqnarray}
holds, implying $E_0^{\rm B}\le E_0^{\rm F}$.
The first part of Theorem~\ref{thm.natural_spinless} is thus proved.
As a simple corollary, the ground-state energies for
a given chemical potential $\mu$
also satisfy Eq.~\eqref{eq.e0b.leq.e0f}.

In order to prove the strict version of the natural inequality,
let us consider $\calL^{S} \equiv \left(\calK^S\right)^n$,
where $S={\rm B,F}$, for a positive integer $n$.
In the occupation number basis, the matrix element of $\calL$ is expanded as
\begin{equation}
\calL^{S}_{ab} = \sum_{c_1,\ldots,c_{n-1}}
\calK^{S}_{ac_1}\calK^{S}_{c_1c_2}\calK^{S}_{c_2c_3}\ldots
\calK^{S}_{c_{n-1}b},
\label{eq.L_elem}
\end{equation}
in which each term in the sum represents a particle hopping process among the connected sites.

From the definition of $\calL^S$ and
the relation~\eqref{eq.Kab}
between $\calK^{\rm B}$ and $\calK^{\rm F}$,
we have the inequality for
matrix elements of $\calL^{\rm B,F}$:
\begin{eqnarray}
\calL^{\rm B}_{ab} &=&\sum_{c_1,\ldots,c_{n-1}}
\calK^{\rm B}_{ac_1}\calK^{\rm B}_{c_1c_2}\calK^{\rm B}_{c_2c_3}\ldots
\calK^{\rm B}_{c_{n-1}b}\\
&=&\sum_{c_1,\ldots,c_{n-1}}
|\calK^{\rm F}_{ac_1}\calK^{\rm F}_{c_1c_2}\calK^{\rm F}_{c_2c_3}\ldots
\calK^{\rm F}_{c_{n-1}b}|\nonumber\\
&\ge& |\sum_{c_1,\ldots,c_{n-1}}
\calK^{\rm F}_{ac_1}\calK^{\rm F}_{c_1c_2}\calK^{\rm F}_{c_2c_3}\ldots
\calK^{\rm F}_{c_{n-1}b}|
=|\calL_{ab}^{\rm F}|.
\end{eqnarray}
This applies, in particular, to the diagonal elements with
$b  = a$.

From Eq.~\eqref{eq.Kab}, the matrix elements of $\calK^{\rm F}$ and
thus the amplitudes of the process in Eq.~\eqref{eq.L_elem} can be
negative for fermions, while they are nonnegative for bosons.
The difference between bosons and fermions shows up
exactly when two particles are exchanged. To make two-particle exchange
process possible, let us introduce  
a ``branching'' site directly connected to three or more sites
belonging to the lattice.
An example of the branching site connected to three sites is shown in
Fig.~\ref{fig.particle-exchange}.
If the number of particle falls in
the range $2 \leq M \leq N-2$, two particles can be exchanged from an
initial state $|\phi^a \rangle$ and back to the same state in $6$
hoppings, with the aid of the branch structure. An example of particle exchange process on a lattice with a
branching site is demonstrated schematically in
Fig.~\ref{fig.particle-exchange}.  The contribution to the diagonal
elements of bosons $\calL_{aa}^{\rm B}$ is always positive at $n=6$, while
the contribution to $\calL_{aa}^{\rm F}$ is negative when two particles
are exchanged.
On the other hand, there is always a positive contribution
to $\calL_{aa}^{\rm B}$ and $\calL_{aa}^{\rm F}$
in the expansion of Eq.~\eqref{eq.L_elem}, at least from
the invariant process $c_j=a$ in which no particle moves
in $n$ steps.
Thus, the strict inequality $\calL_{aa}^{\rm B}>|\calL_{aa}^{\rm F}|$
holds in this case.

When the lattice $\Lambda$ is connected, any basis state $|\phi^a
\rangle_{\rm B}$ can be reached by 
consecutive applications of the
hopping term in $\calK^{\rm B}$, and thus the matrix $\calK^{\rm
B}_{ab}$ satisfies the connectivity.  Together with the property
$\calK_{ab}^{\rm B} \geq 0$, $\calK^{\rm B}_{ab}$ (and thus also
$\calL^{\rm B}_{ab}$) is a Perron-Frobenius
matrix~\cite{horn2012matrix}.
Applying a corollary of the
Perron-Frobenius theorem~\footnote{See,
for example, Theorem 8.4.5 of Ref.~\protect\onlinecite{horn2012matrix}.}
we find $\kappamax^{\rm B} >
\kappamax^{\rm F}$ and hence the latter part of the theorem follows.
\end{proof}

We note in passing that, a consequence of the Perron-Frobenius
theorem is that the ground state of bosons has
a nonvanishing amplitude $_{\rm B}\langle \phi^a | \Psi_0 \rangle_{\rm B}$
with a definite (say, positive) sign for every basis state
$| \phi^a \rangle_{\rm B}$.
This may be understood as a lattice version of the ``no-node''
theorem~\cite{Feynman,Wu09}.

\begin{figure}
\centering
\includegraphics[width=0.6\columnwidth]{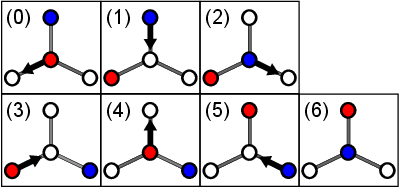}
\caption{A schematic example to show a two-particle exchange process in six steps, where the four-site branch cluster is a subsection of a real arbitrary lattice.}
\label{fig.particle-exchange}
\end{figure}

\subsection{\label{spinful}Natural inequality for spinful case}

Let us now discuss the spinful case.
Here we compare spinful hard-core bosons and spinful fermions
on a finite lattice, with spin-$1/2$.
While actual bosons are known to have only integer spins,
they can have pseudospin-$1/2$, which is sufficient for
the present discussion.
Here the ``hard-core bosons'' means that
two or more particles with the \emph{same}
(pseudo) spin cannot occupy the same site:
$n_{j\sigma}=0,1$, where $\sigma=\uparrow, \downarrow$.
With this constraint, we consider the Hamiltonian,
\begin{eqnarray}
\calH&=&-\sum_{j\ne k}\sum_{\sigma}\left(
t_{jk}c^{\dagger}_{j\sigma}c_{k\sigma}+\mbox{H.c.}  \right)
-\sum_{j\sigma}\mu_j n_{j\sigma}\nonumber\\
{}&&+\sum_{j\ne k}\sum_{\sigma\sigma'}V_{jk}n_{j\sigma}n_{k\sigma'}
+\sum_{j}U_{j}n_{j\uparrow}n_{j\downarrow},
\label{eq.spinfulHam}
\end{eqnarray}
which is a generalization of Eq.~\eqref{eq.Ham} with the
introduction of the spin degrees of freedom $\sigma=\uparrow, \downarrow$.

Let us first discuss the case in which all $U_j$'s are finite.
Then the following simple generalization of
Theorem~\ref{thm.natural_spinless} holds:
\begin{theorem}\textup{(Natural inequality for spinful case with finite $U_j$'s)}

For any set of finite $U_j$'s, if all the hopping amplitudes $t_{jk}$
are real and nonnegative, the inequality~\eqref{eq.e0b.leq.e0f} holds
for any given number of particles $M\le2N$ on a finite lattice $\Lambda$
with $N$ sites. Furthermore, if the lattice $\Lambda$ is connected, and
has a site directly connected to three or more site, and if the number of
particles satisfies $3\le M\le 2N-3$, the strict inequality holds.
\label{thm.natural_spinful}
\end{theorem}
The detailed proof including the restriction of filling, which is a straightforward generalization of the proof of
Theorem~\ref{thm.natural_spinless}, is given in Appendix~\ref{app:proof2}.

Now let us discuss the case $U_j = + \infty$.
The first half of Theorem~\ref{thm.natural_spinful},
the non-strict version of the inequality, remains
unaffected by taking the limit $U_j = + \infty$.
It is easily proved by variational principle in the same manner as in Proof of Theorem~\ref{thm.natural_spinless}.
However, the latter half of Theorem~\ref{thm.natural_spinful},
the strict inequality, is affected by taking the limit.

The proof of the strict inequality is based on the Perron-Frobenius
theorem, which requires the irreducibility of the matrix.
For spinless particles and spinful particles with finite $U_j$'s,
when the lattice is connected,
any pair of occupation number basis states
$|\Phi^{a}\rangle$ and $|\Phi^{b}\rangle$
of the many-particle problem
are connected by consecutive application of particle hoppings.
This implies the irreducibility
of the matrix representing the many-body Hamiltonian.
However, in the case of spinful system with $U_j=+\infty$,
connectivity of the lattice does not guarantee the irreducibility
of the many-body Hamiltonian matrix.
An illustrative example is the Hubbard model with $U_j = +\infty$
at half filling.
Each site is occupied by a particle with either spin up or spin
down; there are many
occupation-number
basis states corresponding to different spin
configurations.
However, since there is no empty site, and double
occupancy with spin up and down particles is forbidden,
each basis state is not connected by hopping
to any other basis state.
Therefore, in order to prove the strict inequality,
we need some additional condition which guarantees
the irreducibility of the Hamiltonian matrix.
In fact, the irreducibility of the Hamiltonian matrix at $U_j=+\infty$,
and application of the Perron-Frobenius theorem
were discussed earlier by Tasaki~\cite{tasaki-1989, tasaki-1998}
in the context of Nagaoka's ferromagnetism.
Nagaoka's ferromagnetism is a mechanism of ferromagnetism
in the Hubbard model with a single hole doped into the half filling with $U_j = + \infty$,
and can be understood as a consequence of the Perron-Frobenius theorem.
For that, the irreducibility of the Hamiltonian matrix in a
certain basis is required. In Ref.~\onlinecite{tasaki-1998},
a sufficient condition for the irreducibility was presented:
if the entire lattice is connected by exchange bonds, then
the Hamiltonian matrix in the occupation number basis
is irreducible. Here ``exchange bond''~\cite{tasaki-1998} is defined by a pair
of sites which belongs to a loop of length three or four,
and the whole lattice remains connected via nonvanishing
hopping amplitudes even when the two sites are removed.
Thus we obtain

\begin{theorem}\textup{(Natural inequality for spinful case not above  
half filling)}

When $U_j$'s are either $+\infty$ or finite,
if all the hopping amplitudes $t_{jk}$
are real and nonnegative, the inequality~\eqref{eq.e0b.leq.e0f} holds
for any given number of particles $M \le N$ on a finite lattice $\Lambda$
with $N$ sites. Furthermore, if the entire lattice $\Lambda$ is
connected by exchange bonds, and if the number of
particles satisfies $3 \le M \le N-1$, the strict inequality holds.
\label{thm.natural_infiniteU}
\end{theorem}

The outline of the proof of Theorem~\ref{thm.natural_infiniteU}  including the restriction of filling, and the numerical verification of the theorems are presented in Appendix~\ref{app:proof3}.

In summary, in this section we have presented three theorems for the validity of the natural inequality for spinless and spinful cases, respectively. Although the proofs of the sufficient conditions for strict version of the natural inequality (see Appendix~\ref{app:proof2},\ref{app:proof3}) are somewhat more involved, the basic idea behind the proofs is the same as in that for Theorem~\ref{thm.natural_spinless}.
That is, bosons have a strictly lower ground-state energy than fermions, when the hopping amplitudes are non-negative and an exchange of particles is allowed.

\section{\label{frustration}Unified understanding of frustration
and diamagnetic inequality}

The role played by \emph{frustration} is of central importance in the
proofs of the theorems. The terminology ``frustration''
is often used for antiferromagnetically interacting spin system
on geometrically frustrated lattices, such as triangular, kagome and pyrochlore
lattices. When there is no global state of the system that minimizes
every antiferromagnetic interaction, there is some frustration.
More generally, frustration may be applicable
to a system with competing interactions,
when the ground state does not minimize individual
interaction simultaneously~\cite{Diep-FrustratedSpinSystems}.

To see that the sign of hopping amplitudes $t_{jk}$ in a many-boson system
is related to frustration, it is illuminating to map
the hard-core boson problem to a spin-$1/2$
quantum spin system~\cite{Matsubara_Matsuda1956}.
The mapping is based on the equivalence between hard-core boson
operators and spin-$1/2$ operators:
\begin{align}
 S_j^+ & \sim b_j^{\dag} , & S_j^- & \sim b_j ,
& S_j^z & \sim b_j^{\dag}b_j - \frac{1}{2} .
\end{align}
It is then easy to see that a hopping term for hard-core bosons
maps to an in-plane exchange interaction:
\begin{align}
 - t_{jk} \left( b^{\dag}_j b_k + b^{\dag}_k b_j \right)
\sim  J^{\perp}_{jk} \left( S^x_j S^x_k + S^y_j S^y_k \right) ,
\end{align}
where $J^{\perp}_{jk} = - 2 t_{jk}$.
Thus the nonnegative $t_{jk}$ corresponds to
ferromagnetic exchange interaction, in terms of the spin system.
When all the exchange couplings are ferromagnetic, there
is no frustration. Namely, every in-plane exchange
interaction energy can be minimized simultaneously
by aligning all the spins to the same direction in the $xy$-plane.
Going back to the original problem of quantum particles,
the direction of the spins in the $xy$-plane corresponds
to the quantal phase of particles at each site.
If all the hopping amplitudes are nonnegative,
every hopping term can be simultaneously minimized
by choosing a uniform phase throughout the system.
In this sense, bosons with nonnegative hopping amplitudes
are unfrustrated with respect to their quantal phase.

Let us now consider the case of fermions.
Since Fermi statistics brings in negative signs
even if all the hoppings $t_{jk}$ are nonnegative,
it would be natural to expect that Fermi statistics
leads to some kind of frustration.
However, it is difficult to formulate this
based on the above mapping to an $S=1/2$ spin system.
To understand the frustration induced by Fermi statistics
in many-particle systems, we introduce an alternative
mapping of the many-body Hamiltonian into
a single-particle tight-binding model.
That is, we identify each of the many-body
occupation number basis states $|\Phi^{a}\rangle$
with a site on a fictitious lattice.
If two occupation number basis states
$|\Phi^a\rangle$ and $|\Phi^b\rangle$ are connected by
Hamiltonian, $\langle \Phi^b|\calH|\Phi^a\rangle\ne 0$, there is a link
connecting sites $a$ and $b$ in the fictitious lattice. If we can start from
an initial state, and return back to the same state by successive
applications of the Hamiltonian~\eqref{eq.Ham}, there is a loop in the
fictitious lattice. For bosons, there is no extra phase in the loop. In
other words, the fictitious lattice for hard-core bosons is flux
free. Therefore, there is no frustration for bosons because there is a
constructive interference among all the paths. In contrast, for
fermions, in the original many-body problem,
if two particles are exchanged and the system returns back to the
initial state, the system acquires an extra $\pi$ phase. The minus sign introduced by Fermi statics is relevant to sign structure~\cite{sign_structure}.
Upon the mapping to the single-particle problem, this is equivalent
to the presence of a $\pi$-flux in the corresponding loop
in the fictitious lattice.
This can be interpreted as frustration, which causes
destructive interferences among different paths.

For a single-particle tight-binding model, introduction of a flux always
raises or does not change the ground-state energy, which is known
as diamagnetic inequality~\cite{diamagnetic_inequality}.
The first half of Theorem~\ref{thm.natural_spinless},
which states the non-strict inequality, may be then regarded
as a corollary of the diamagnetic inequality.
On the other hand, the latter half of the Theorem~\ref{thm.natural_spinless}
concerning the strict inequality does not, to our knowledge,
follow from known results on the diamagnetic inequality.
In fact, the arguments in the proof of Theorem~\ref{thm.natural_spinless}
can be applied to a strict version of the diamagnetic inequality
on general lattices. The general result can be summarized
as follows.
\begin{theorem}\textup{(General diamagnetic inequality and its strict version)}

Let us consider a single particle on a finite lattice $\Xi$,
with the eigenequation
\begin{equation}
- \sum_{\beta \in \Xi}
\tau_{\alpha \beta} \psi_\beta = E \psi_\alpha .
\end{equation}
In general, $\tau_{\alpha \beta}$ is complex, with
$\tau_{\alpha\beta} = \tau^*_{\beta\alpha}$.
The ground-state energy $E_0$ for a given set of the hopping
amplitudes $\{ \tau_{\alpha\beta} \}$ satisfies
\begin{equation}
  E_0(\{ \tau'_{\alpha\beta} \equiv
|\tau_{\alpha\beta}| \}) \leq E_0(\{ \tau_{\alpha\beta} \}) .
\end{equation}
Furthermore, the strict inequality,
\begin{equation}
  E_0(\{ \tau'_{\alpha\beta} \equiv |\tau_{\alpha\beta}| \})
< E_0(\{ \tau_{\alpha\beta} \})
\label{eq.strict_dia}
\end{equation}
holds, provided that the lattice $\Xi$ is connected
and there is at least one loop which contains a nonvanishing flux.
A sequence of sites $\{ \alpha_0, \alpha_1, \alpha_2,\ldots,  \alpha_n \}$,
which satisfies $\alpha_l \neq \alpha_{l+1}$,
$\tau_{\alpha_l \alpha_{l+1}} \neq 0$ and
$\alpha_n = \alpha_0$ is called a loop.
The loop contains a nonvanishing flux when the product
\begin{equation}
 \tau_{\alpha_0 \alpha_1} \tau_{\alpha_1 \alpha_2}
\tau_{\alpha_2 \alpha_3} \ldots \tau_{\alpha_{n-1} \alpha_n}
\label{eq.product}
\end{equation}
is not positive (either negative or not real).
\label{thm.diamagnetic}
\end{theorem}

The non-strict version is the standard diamagnetic
inequality~\cite{diamagnetic_inequality,Barry_Simon_diamag}.
However, the strict inequality obtained here appears new,
also in the general context of diamagnetic inequality. 
The detailed proof of Theorem~\ref{thm.diamagnetic} can be found in Appendix~\ref{app:proof4}.

Mapping of the original quantum many-particle problem
to the single-particle problem on a fictitious lattice
provides a unified understanding of
frustration of quantal phase.
When there is a nonvanishing flux in the original
many-particle problem, we observed that there is a frustration
among local quantal phases, which we call hopping frustration.
On the other hand, when the particles in the original
problem are fermions, there is also a frustration
among quantal phases introduced by the Fermi statistics,
which we name \emph{statistical frustration}.
In the original many-particle problem,
the statistical frustration appears rather different
from the hopping frustration.
However, upon mapping to the single-particle problem
on the fictitious lattice, both hopping frustration
and statistical frustration are represented by a
nonvanishing flux in the fictitious lattice.
This provides a unified understanding of hopping
and statistical frustrations.

A system of many bosons with only nonnegative hopping amplitudes
$t_{jk}$ are free of frustration.  
Introduction of any frustration into such a system, for example magnetic flux (hopping frustration), is expected
not to decrease the ground-state energy.
This is a lattice version of Simon's
universal diamagnetism of bosons~\cite{Barry_Simon_diamag}.
However, in many-fermion system, where the statistical frustration
exists, the effect of introducing hopping frustration is a nontrivial problem.
In such a case, the ground-state energy may or may not decrease,
depending on the system in the question.
That is, diamagnetism is not universal
in spinless fermion systems.
Correspondingly, the orbital magnetism of fermions can be
either paramagnetic or diamagnetic,
depending on the model~\cite{Vignale-OrbitalPara-PRL1991}.
Considering each of the frustrations introduces a particular
pattern of magnetic flux in the fictitious lattice,
it is certainly possible that in some cases the hopping frustration
may (partially) cancel the effect of statistical frustration,
so that the introduction of the hopping frustration
actually decreases the ground-state energy. This reveals the fact that the natural inequality could be violated by the introduction of hopping frustration. Some concrete examples, in which the natural inequality is violated, are demonstrated in the following Section.

\section{\label{reversed}Violation of the natural inequality}

In the following, we discuss how the natural inequality can be
violated. Theorems~\ref{thm.natural_spinless}
and~\ref{thm.natural_spinful} leave the possibility of
violation of the inequality in the presence of a hopping frustration,
that is, by choosing negative or complex hopping amplitudes $t_{jk}$.
However, the hopping frustration is a necessary but
not sufficient condition to reverse the natural inequality.
We will demonstrate that the violation of natural inequality
indeed happens in several frustrated systems.
For simplicity, we limit ourselves to the comparison
between spinless fermions and hard-core bosons, with
no interaction other than the hard-core constraint.
The case with density-density interaction will be discussed at the end
of this section.

\subsection{Particles on a ring}

\label{sec:ring}
We start with the best understood and solvable model in one
dimension:
\begin{equation}
\calH=-\sum_{j=1}^N (c_{j}^{\dag}c_{j+1}+\textrm{H.c.}).
\end{equation}
The hard-core boson version of this model, which is equivalent to the
spin-$1/2$ $XY$ chain, can be mapped to free fermions on a
ring by Jordan-Wigner transformation~\cite{LSM, S_Katsura}.
Thus energy eigenvalue problem of hard-core bosons and fermions
on a ring are almost the same, except for the subtle
difference in the boundary condition. For the periodic or antiperiodic
boundary conditions $c_{N+1} \equiv \pm c_1$, the Jordan-Wigner fermions
$\tilde{f}_j$ obey the boundary condition
$ \tilde{f}_{N+1} = \mp e^{i\pi M} \tilde{f}_1$,
where $M$ is the number of Jordan-Wigner fermions
(equals to the number of bosons). If $M$ is assumed as even, it implies
that hard-core bosons with the periodic (antiperiodic) boundary
condition is mapped to free fermions with the antiperiodic (periodic,
respectively) boundary condition. 

Now let us discuss the dependence of the ground-state energy on the
boundary condition. Assuming $M=N/2$ is even, the ground-state energy
density (ground-state energy per site) is given as
\begin{equation}
\epsilon_0 =\frac{E_0}{N}=- \frac{2}{N} \sum_{k} \cos{k}, \label{eq:GS}
\end{equation}
where $k$ is taken over all the momenta in the Fermi sea, $ - \pi/2 \leq
k < \pi/2$.  For the periodic boundary condition (PBC), the wavenumber $k$ is
quantized as $k = 2\pi n/N$, while $k= \pi (2n+1)/N$ for the
antiperiodic boundary condition (APBC), where $n$ ($ - N/4 \leq n < N/4$) is an
integer.

The ground-state energy density asymptotically converges, in the
thermodynamic limit $N \to \infty$, to the same integral for either
boundary condition.  Nevertheless, it does depend on the boundary
condition for a finite $N$.  The difference of ground-state energy is
exactly calculated as
\begin{equation}
\frac{E_0^{\mbox{\scriptsize PBC}}}{N}-\frac{E_0^{\mbox{\scriptsize
APBC}}}{N} =\frac{2[1-\cos(\pi/N)]}{N\sin(\pi/N)}>0,
\end{equation}
for any $N>1$.  The antiperiodic boundary condition gives the lower
ground-state energy.  The leading order of difference can be extracted
in the limit of large $N$ as,
\begin{align}
  \frac{E_0^{\mbox{\scriptsize PBC}}}{N} & = - \frac{2}{\pi} +
\frac{2\pi}{3N^2} + \frac{2\pi^3}{45N^4} + O(\frac{1}{N^6}) \\
\frac{E_0^{\mbox{\scriptsize APBC} }}{N} & = - \frac{2}{\pi} -
\frac{\pi}{3N^2} - \frac{7\pi^3}{180N^4} + O(\frac{1}{N^6}),
\end{align}
for the periodic and antiperiodic boundary conditions.  The
leading term of $O(1/N^2)$ is also determined by conformal
field theory\cite{Ginsparg-Applied-CFT,Alcaraz-CFT}.  It can be seen
that the noninteracting fermions on a ring have a lower ground-state energy with the
antiperiodic boundary condition.

As a result, with periodic boundary condition, hard-core bosons have a
lower ground-state energy than fermions, in full agreement with
Theorem~\ref{thm.natural_spinless}.
On the other hand, the ground-state energy of
hard-core bosons is higher than that of fermions with anti-periodic
boundary condition. The anti-periodic boundary condition can be
understood as a result of insertion of $\pi$-flux inside the ring.
This hopping frustration cancels the
statistical frustration so that the natural inequality is violated.

This example of tight-binding model may look trivial, and indeed the
calculation itself has been known for years. Nevertheless, it is very
useful in highlighting the central physics of the problem, that is, the
effect of the statistical frustration of fermions can be canceled by the flux
or hopping frustration.
The present finding can also be applied to construction of
more nontrivial examples, as we will discuss in the Sec.~\ref{sec:coupled-ring}.

\subsection{Coupled rings}
\label{sec:coupled-ring}
Since hard-core bosons have a higher ground-state energy than fermions on
a ring containing $\pi$ flux inside the ring as proved in Sec.~\ref{sec:ring},
we can construct a series of systems where $E_0^{\rm B} > E_0^{\rm F}$,
by taking many such small rings and connecting them with weak
hoppings. If the inter-ring hoppings are weak enough, they are expected
not to revert the inequality and $E_0^{\rm B} > E_0^{\rm F}$ would be
kept~\cite{Chamon-priv}.

We prove rigorously that, the reversed natural inequality is
indeed still kept
in coupled $\pi$-flux rings, connected by weak hoppings, even in the
thermodynamic limit. One example is $\pi$-flux octagon-square
model. The lattice structure is shown in Fig.~\ref{fig.octagon_square}
(a), where one unit cell is shown in green with basis vectors
$\vec{a}_1=(3,0)$ and $\vec{a}_2=(0,3)$.
This lattice can be deformed into the (topologically equivalent)
$\frac{1}{5}$-depleted square lattice\cite{Troyer1996,Sachdev1996},
which is known for the model of the quasi two-dimensional compound $\textrm{CaV}_4\textrm{O}_9$.
Thus the octagon-square lattice is also called as deformed $\frac{1}{5}$-depleted square lattice. It is sometimes also called as decorated square lattice~\cite{depleted-square2, depleted-square1}.
The hopping amplitudes on thick
and broken lines are denoted by $t$ and $t'$, respectively. The
Hamiltonian is given by
\begin{equation}
\calH=-t \sum_{\langle i,j\rangle\in\textrm{thick,oriented}}
e^{i\pi/4}c_i^{\dag}c_j -t'\sum_{\langle
i,j\rangle\in\textrm{broken}}c_i^{\dag}c_j+\textrm{H.c.},
\label{eq.Haml_octagon_square}
\end{equation}
where ``thick, oriented'' and ``broken'' refer respectively to
the links drawn with arrows and those drawn as broken lines
in Fig.~\ref{fig.octagon_square}(a).
We also assume $t>t'>0$.

By the choice of $e^{i\pi/4}$ hopping phase on the
oriented thick lines, there is a $\pi$ flux in every square. Therefore,
it can be regarded as a model of coupled $\pi$-flux rings by weak
hopping $t'$.
In order to prove $E_0^{\rm B} > E_0^{\rm F}$ rigorously in the coupled rings,
we seek a lower bound for $E_0^{\rm B}$ and an upper bound for $E_0^{\rm F}$.
If the former is higher than the latter, the desired
inequality is proved.
We introduce the positive semi-definite operators,
\begin{eqnarray}
A= t' \sum_{\langle i,j\rangle\in\textrm{Broken}}
(c_i^{\dag}+c_j^{\dag})(c_i+c_j)\ge 0, \label{eq.operatorA}\\ B= t'
\sum_{\langle i,j\rangle\in\textrm{Broken}}
(c_i^{\dag}-c_j^{\dag})(c_i-c_j)\ge 0, \label{eq.operatorB}
\end{eqnarray}
where $A\ge0$ means $\langle \Phi | A | \Phi\rangle\ge0$ for any state
$|\Phi\rangle$.  Therefore, the Hamiltonian for fermions and
bosons can be written as
\begin{eqnarray}
\calH^{\rm F} = \tilde{\calH}^{\rm F}-A=\sum_{\Diamond}
h_{\Diamond}^{\rm F}-A,\\
\calH^{\rm B} = \tilde{\calH}^{\rm B}+B=
\sum_{\Diamond} h_{\Diamond}^{\rm B}+B,
\label{eq.HB.decomp}
\end{eqnarray}
where $h_{\Diamond}^{\rm F}=-t\sum_{i=1}^{4}
(e^{i\pi/4}c_i^{\dag}c_{i+1}+\textrm{H.c.})+t'\sum_{i=1}^4
c_i^{\dag}c_i$ and $h_{\Diamond}^{\rm B}=-t\sum_{i=1}^{4}
(e^{i\pi/4}c_i^{\dag}c_{i+1}+\textrm{H.c.})-t'\sum_{i=1}^4
c_i^{\dag}c_i$, the cluster Hamiltonians defined on a solid-line square
for fermions and bosons, respectively. Noticing $h_{\Diamond}$ commutes
with each other, the ground-state energy of $\tilde{\calH}$ is simply
given by the summation~\cite{Anderson}:
\begin{equation}
\tilde{E}_0= \sum_{\Diamond_i}\epsilon_{\Diamond_i},
\end{equation}
where $\tilde{E}_0$ and $\epsilon_{\Diamond_i}$ are the ground-state
energy of $\tilde{\calH}$ and that of $h_{\Diamond_i}$ on $i$-th $\pi$-flux
square, respectively.

\begin{figure}
\centering
\subfigure[]{
\includegraphics[width=1.6in]{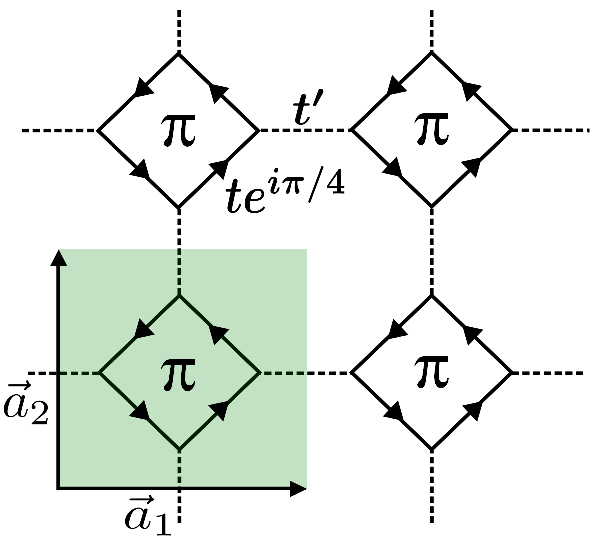}}
\subfigure[]{
\includegraphics[width=1.6in]{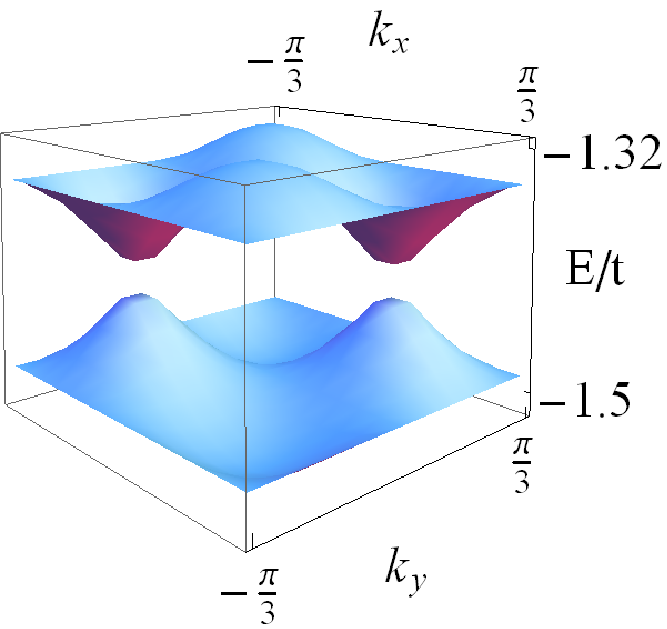}}
\caption{(a) $\pi$-flux octagon-square lattice,
in which a unit cell is shown in green. (b) The lowest two bands of Hamiltonian~\eqref{eq.Haml_octagon_square}
with $t=1$, $t'=0.1$.}
\label{fig.octagon_square}
\end{figure}

Because the operators $B$ is positive semi-definite, the ground-state
energy of bosons satisfies
\begin{equation}
E_0^{\rm B}=\langle \Phi|\calH^{\rm B}|\Phi\rangle \ge\langle
\Phi|\tilde{\calH}^{\rm B}|\Phi\rangle \ge
\tilde{E}_0^{\rm B}=\sum_{\Diamond_i}\epsilon_{\Diamond_i}^{\rm B},
\label{eq.lower_bound_octagon_square}
\end{equation}
where $|\Phi\rangle$ is assumed as the ground state of $\calH^{\rm B}$.

On the other hand, an upper bound of fermions can be derived as,
\begin{equation}
E_0^{\rm F}=\langle \Psi|\calH^{\rm F}|\Psi\rangle \le \langle
\tilde{\Psi}|\calH^{\rm F}|\tilde{\Psi}\rangle \le \langle
\tilde{\Psi}|\tilde{\calH}^{\rm F}|\tilde{\Psi}\rangle =\tilde{E}_0^{\rm
F}=\sum_{\Diamond_i}\epsilon_{\Diamond_i}^{\rm F},
\end{equation}
where $|\Psi\rangle$ and $|\tilde{\Psi}\rangle$ are the ground states of
$\calH^{\rm F}$ and $\tilde{\calH}^{\rm F}$, respectively.

By exact diagonalization, we obtain the ground-state energies
$\epsilon_{\Diamond}^{\rm B, \rm F}(m)$ in given $m$ particles
sectors, shown in Table~\ref{table.octagon-square} in Appendix~\ref{app:tables}.
The number of unit cells is assumed as $N$. From the results of exact
diagonalization, a lower bound for bosons is given by $E_0^{\rm B}\ge
-2N(t+t')$ when $t'/t\le 2-\sqrt{2}$, or $E_0^{\rm B}\ge
-N(\sqrt{2}t+3t')$ when $2-\sqrt{2}<t'/t<1$.  An upper bound for
fermions is given by the $\tilde{E}_0^{\rm F}$, which is dependent on
the density pattern on the whole lattice.  At half filling, an upper
bound of fermions is obtained as
\begin{equation}
E_0^{\rm F}\le -2N(\sqrt{2}t-t').
\label{eq.E0F_upper}
\end{equation}
Thus, when the ratio falls in this range $t'/t < (\sqrt{2}-1)/2$, we
have $E_0^{\rm B} > E_0^{\rm F}$.

Instead of searching an upper bound of fermions, the ground-state energy
of fermions can be exactly calculated at certain filling. For
convenience, $t$ is
set equal to $1$. In the single particle sector, the exact
dispersion relations are obtained by Fourier transformation:
\begin{eqnarray}
E_{\pm}^{(1)}=\pm
\sqrt{(t')^2+2-2t'\sqrt{1-\sin{(3k_x)}\sin{(3k_y)}}},\nonumber\\
E_{\pm}^{(2)}=\pm
\sqrt{(t')^2+2+2t'\sqrt{1-\sin{(3k_x)}\sin{(3k_y)}}},\nonumber
\end{eqnarray}
where $(k_x,k_y)$ is the wavenumber which belongs to the reduced
Brillouin zone $-\pi/3\le k_{x,y}<\pi/3$. The ground-state energy of
fermions at $\mu=0$, which corresponds to the half filling, is given as
\begin{equation}
E_{0}^{\rm F}=\sum_{k_x,k_y}\big[E_{-}^{(1)}(k_x,k_y)+E_{-}^{(2)}(k_x,k_y)\big].
\end{equation}
Under the assumption that the lattice is of size $9L^2$, the number of unit cells $N$ equals
$L^2$. In the thermodynamic limit $L\to \infty$, the ground-state energy
of fermions per unit cell at half filling is given by the integral of
the lowest two bands (shown in Fig.~\ref{fig.octagon_square} (b)) in the
reduced Brillouin zone,
\begin{eqnarray}
\frac{E_0^{\rm F}}{N}&\!=\!&-
\int_{-\pi}^{\pi}\!\!\frac{d\tilde{k}_x}{2\pi}
\int_{-\pi}^{\pi}\!\!\frac{\tilde{k}_y}{2\pi}
\Big[\sqrt{(t')^2+2+2t'\sqrt{1-\sin{\tilde{k}_x}\sin{\tilde{k}_y}}}\nonumber\\
&&{}+\sqrt{(t')^2+2-2t'\sqrt{1-\sin{\tilde{k}_x}\sin{\tilde{k}_y}}}\Big].
\label{eq.octagon_square_Ef}
\end{eqnarray}
It is easily verified that the reversed natural inequality holds with
small ratio of $t'/t$, by comparison of the lower bound of bosons and
numerical integral of Eq.~\eqref{eq.octagon_square_Ef} with given
value of $t'$. For example when $t=1$ and $t'=0.1$, $E_0^{\rm B}\ge
-2.2N > E_0^{\rm F}=-2.831967N$.  When $t'=0.4$, $E_0^{\rm B}\ge -2.8N >
E_0^{\rm F}=-2.885971N.$
The exact result is of course consistent with the
rigorous upper bound~\eqref{eq.E0F_upper}.

Our conjecture that the reversed inequality is kept in the coupled $\pi$-flux rings with weak enough inter-ring hopping is now verified in coupled-square lattice. Moreover, the validity of the conjecture should not depend on the specific lattice. As another example, a proof of the reversed inequality for the breathing kagome lattice at certain filling, which can be regarded as an realization of a coupled-triangle lattice, is presented in  Appendix~\ref{app:coupled_triangles}.

\subsection{System with flux in $2$D and $3$D}
\label{sec:2D-flux}

\begin{figure}
\centering
\subfigure[]{
\label{fig:spectrum28}
\includegraphics[width=2.5in]{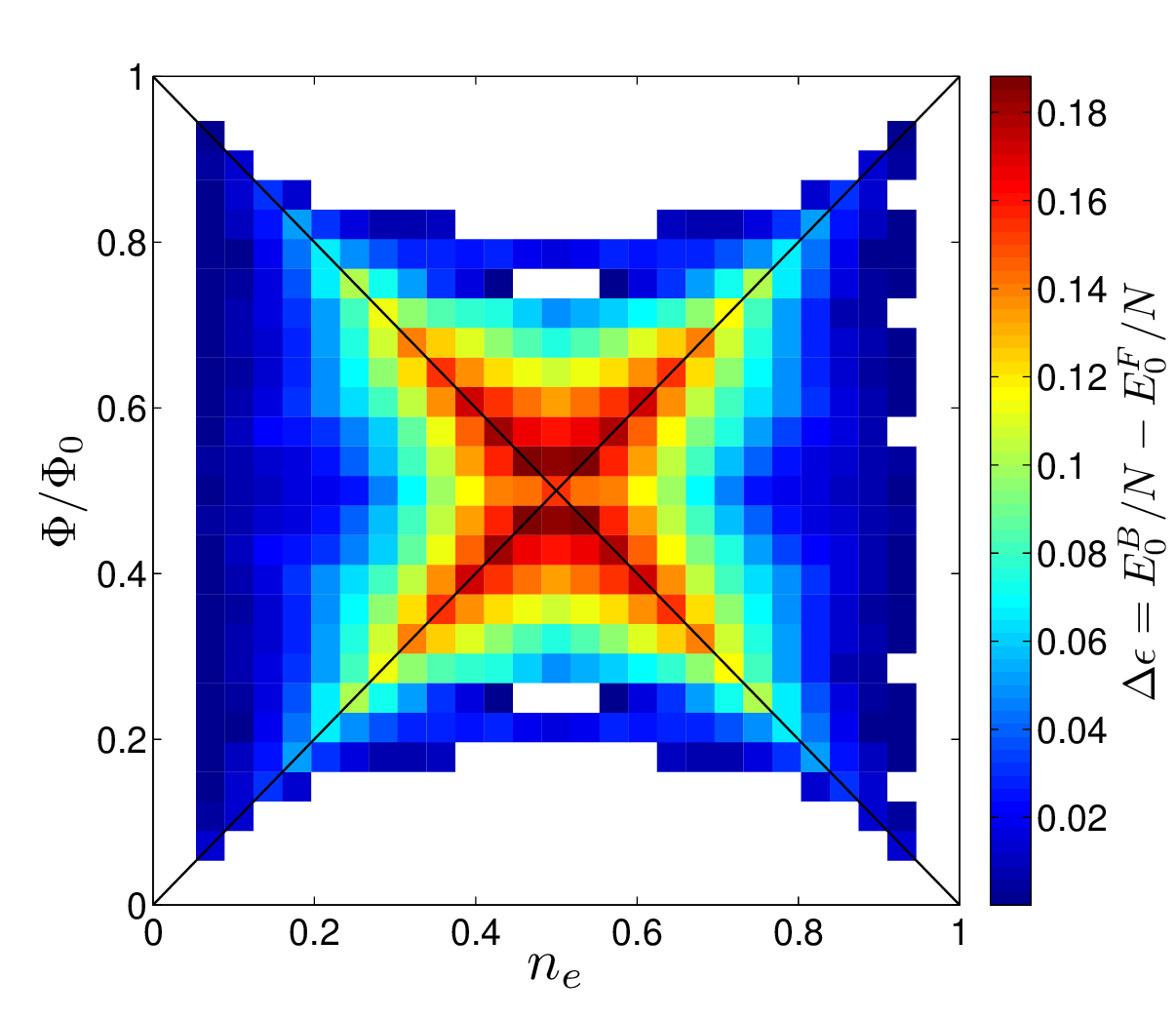}}
\subfigure[]{
\label{fig:spectrum30}
\includegraphics[width=2.5in]{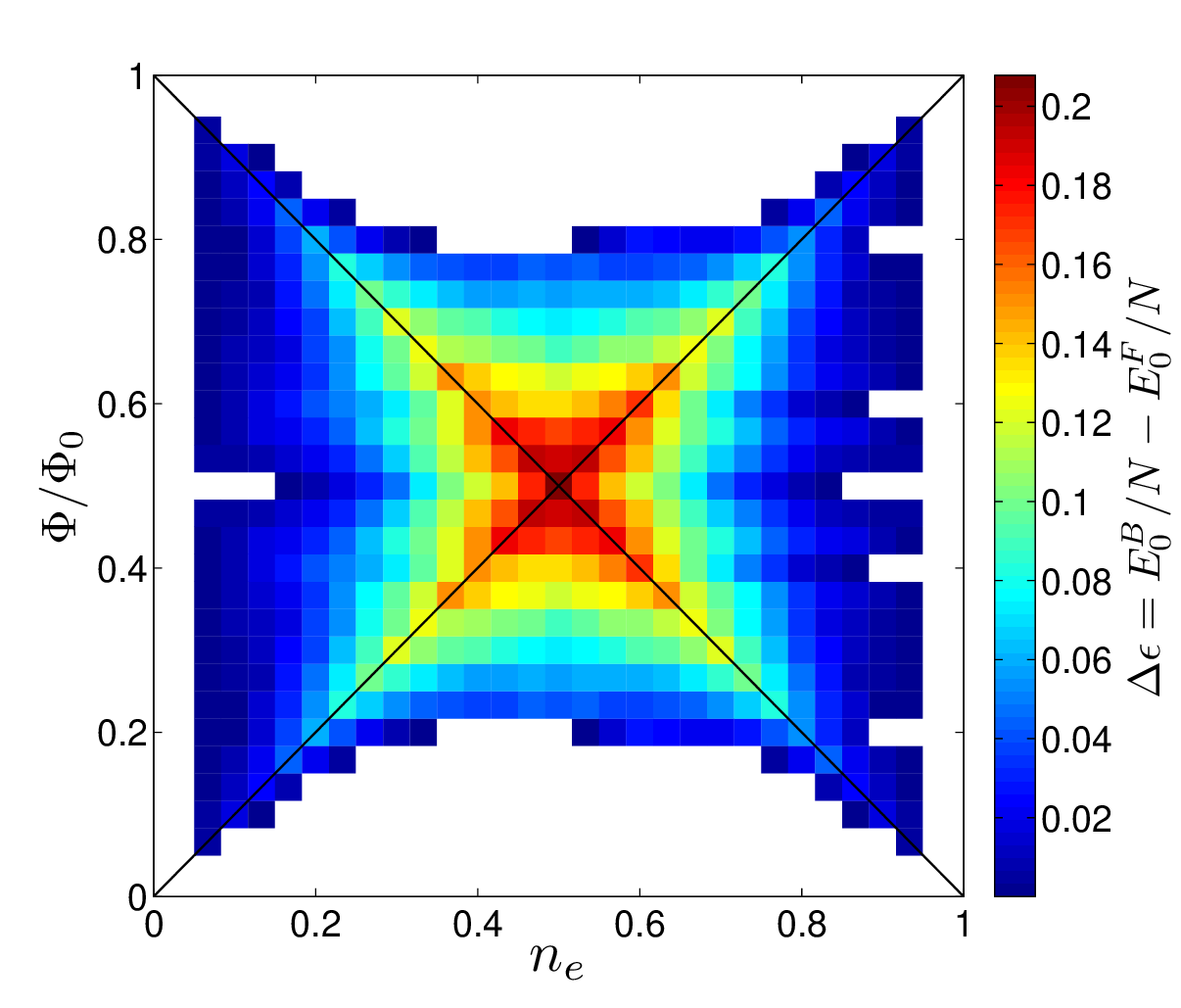}}
\caption{The energy density difference $\Delta \epsilon=E_0^{\rm
B}/N-E_0^{\rm F}/N$ between bosons and fermions on (a) $4\times 7$ and
(b) $5\times 6$ square lattices, where $n_e$ is the number of
particle per site and $\Phi/\Phi_0$ is the number of the flux quanta per
plaquette.}
\label{fig:spectrum}
\end{figure}

As we discussed in Sec.~\ref{sec:ring}, the energy difference between
bosons and fermions on a ring is due to finite-size effect, and
indeed vanishes in the thermodynamic limit.
This is rather natural, it is only the entire system as a ring
that contains $\pi$ flux.
As a simple extension of the idea, here
we consider the two-dimensional square lattice in a
uniform magnetic field, described by the Hamiltonian:
\begin{equation}
\calH =-\sum_{\langle
j,k\rangle}\big(t_{jk}c_{j}^{\dag}c_{k}+\textrm{H.c.}\big),
\label{eq.2D_flux}
\end{equation}
where $t_{jk}=t \exp(i\Phi_{jk}/\Phi_{0})$ and $t>0$.  The flux passing
through every plaquette is $\sum_{\Box}\Phi_{jk}=\Phi$.  With periodic
boundary condition, the total flux is quantized as an integral
multiple of flux quantum ($\Phi_{0}= hc/e$ is $2\pi$ in our unit).
The magnetic field introduces frustration, through the
existence of complex hopping amplitudes $t_{jk}$.
To investigate all the possible values of flux per plaquette,
string gauge~\cite{Hatsugai-stringgauge} is employed.
The string gauge is constructed as follows.
First we choose (the center of) an arbitrary plaquette $S$ as the origin, and draw an oriented path (arrow) from the origin $S$ to every other plaquette. Each oriented path consists of straight segments connecting the centers of neighboring plaquettes.
Once such paths are constructed, the vector potential on each link is
set to $2 \pi m n/N$, where $m$ is the total number of arrows cutting the
edge from the left to the right with respect to the direction of hopping,
and $n$ is an arbitrary integer satisfying $0 \leq n < N$.
Since one of the arrows terminates in each plaquette, 
the flux piercing the plaquette is then $\Phi = n \Phi_0 /N$.
At the origin $S$, where $N-1$ arrows flow from, the flux appears to be
$\Phi = - n (N-1) \Phi_0/N$ instead. However, this is equivalent to
$\Phi = n \Phi_0 /N$, since the flux per plaquette is defined only modulo
$\Phi_0$. In this way, the uniform flux $n \Phi_0/N$ is realized in every
plaquette using the string gauge, although the vector potential is generally
not uniform (translation invariant).

By exact diagonalization, the ground-state energies of bosons and
fermions are obtained with different particle densities ($n_e=M/N$, where $M$ is the number of particles) and various
values of flux.
The relative difference of the ground-state energies in the $4\times 7$ and $5\times 6$ lattices are shown in Fig~\ref{fig:spectrum}.
Here the ground-state energy density differences between bosons and fermions is shown color-coded in the two-dimensional parameter space of the particle density $n_e$ and flux density $\Phi/\Phi_0$.
The natural inequality holds in white regions, while it is violated in colored regions. 
It should be noted that the violation is not necessarily related to band topology. In fact, in the entire region of the parameter space except for $\Phi=0$, each of the single particle bands are characterized by a non-vanishing Chern number~\cite{TKNN}. Nevertheless, the violation of the natural inequality does not happen everywhere. 
Instead, as shown in Fig.~\ref{fig:spectrum}, the violation is nontrivially related to particle density or filling fraction. (Nontrivial dependence on the filling is also found in other models discussed in other sections). To understand the physical origin of the filling-dependence of the relative ground-state energy, one can recall statistical transmutation~\cite{Semenoff,Fradkin} via a flux attachment.  When $\Phi/\Phi_0=n_e$, the background magnetic field can be effectively absorbed by attaching one flux quantum to each particle, at the mean field level ignoring quantum fluctuations. The flux attachment transforms fermions into bosons and vice versa.
In this picture, along the diagonal lines in the plot where
$\Phi/\Phi_0 = n_e$ holds, fermions and hard-core bosons in the magnetic field is mapped respectively to hard-core bosons and fermions in zero field.
According to Theorem 1, the hard-core bosons have a lower ground-state energy
than fermions in zero field. It is thus implied that the violation of the
natural inequality would occur along the diagonal lines.
It should be noted that the flux attachment argument is not rigorous and its range of validity is not established.
Nevertheless, it is remarkable that our numerical calculation indeed reveals the strongest violation along the diagonal lines, as expected from the naive flux attachment argument.

The effect of filling can also be understood in a different way: the energy levels of free electrons (without a lattice or a periodic potential) in a uniform magnetic field are quantized into Landau levels, which can be regarded as completely flat bands.
In the presence of the lattice, each Landau levels are split into dispersive subbands. Nevertheless, one may still regard them as descendants of the Landau level with small dispersion.
Since the main ``disadvantage'' of fermions for lowering the ground-state energy is the Pauli exclusion principle which force some of the fermions to occupy higher-energy states, less dispersive bands are helpful to reverse the natural inequality.(This mechanism will be discussed more explicitly in Sec.~\ref{sec:flat-band}).
The filling $n_e = \Phi/\Phi_0$ corresponds to completely filling the lowest Landau level, and thus can be advantageous to reverse the natural inequality.

We note in passing that, although our numerical results in Figs.~\ref{fig:spectrum} appear almost particle-hole symmetric, a careful examination shows that it is not exactly particle-hole symmetric.
This is because the finite-size lattices used in our calculations are not bipartite, due to the limitation of the system sizes in the exact diagonalization calculation; the bipartiteness is needed for the fermion system on a finite lattice to possess the particle-hole symmetry.

\begin{figure}[htbp]
\centering
\includegraphics[width=3.0in]{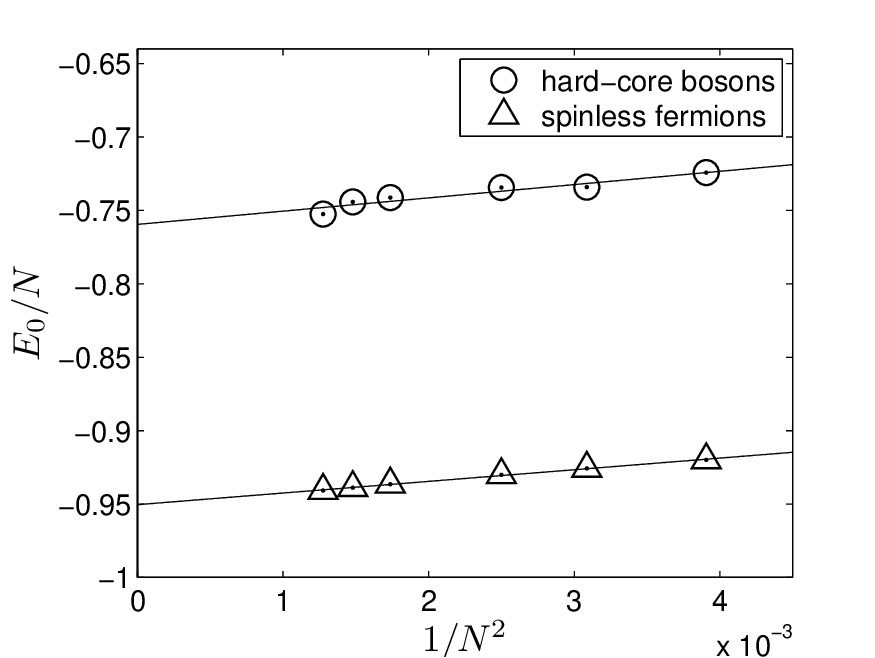}
\caption{Finite-size scaling of ground-state energies in two-dimensional
square lattice with $(N/2-1)\Phi_0/N$ flux per plaquette at filling
fraction $(N/2-1)/N$.  The fitting functions are $E_0^{\rm B}/N =
-0.7593+8.973/N^2+O(N^{-4})$ for hard-core bosons and $E_0^{\rm F}/N
=-0.9507+8.043/N^2+O(N^{-4})$ for fermions respectively.  The
extrapolated ground-state energy density for fermions matches well with
the exact result $-0.958091$ in Eq.~\eqref{eq.exact}.}
\label{fig.scaling_half_filling}
\end{figure}

We plotted Fig.~\ref{fig.scaling_half_filling} to show the finite-size
scalings. Figure~\ref{fig.scaling_half_filling} is the finite-size scaling
with $(N/2-1)\Phi_0/N$ flux per plaquette near half filling
$(N/2-1)/N$. The exact half filling on finite-size lattices ($N/2$
particles on $N$ sites) and the corresponding $\Phi_0/2$ flux per
plaquette are avoided to reduce the strong finite-size effect
(oscillatory behavior) due to commensuration, while the extrapolation
corresponds to the half filling in the thermodynamic limit. The
extrapolation suggests that the fermions have a lower ground-state energy
in the thermodynamic limit.
Actually, we can prove~\cite{ourPRL} rigorously
in the following that this is indeed the case.

\begin{figure}
\centering
\subfigure[]{
\includegraphics[width=1.4in,height=1.5in]{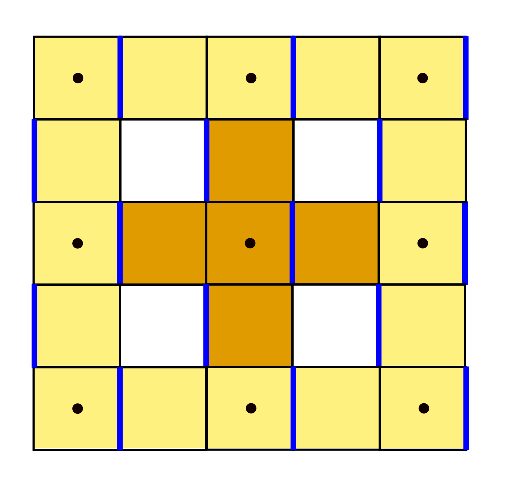}}
\subfigure[]{
\includegraphics[width=1.8in,height=1.5in]{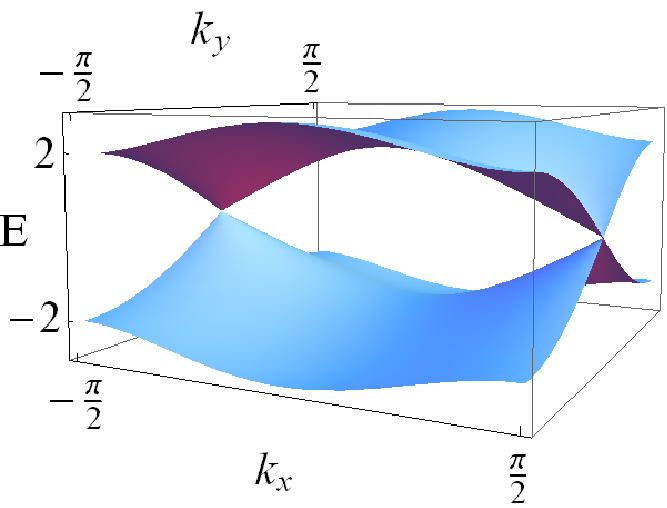}}
\caption{(a) The square lattice with $\pi$ flux in each plaquette.
The brown cross represents a cluster of $12$ sites.
(b) The energy bands in the first Brillouin zone.}
\label{fig.pi_flux}
\end{figure}

As proved by Lieb~\cite{Lieb-flux}, the optimal energy minimizing flux is $\pi$ per plaquette for square lattice at half filling.
Let us discuss the square lattice with $\pi$-flux per plaquette, described by the Hamiltonian~\eqref{eq.2D_flux}.  For convenience, we
choose the gauge so that the hopping amplitude $t_{jk}$ is $+1$ on the black links, and $-1$ on the blue ones as shown in Fig.~\ref{fig.pi_flux} (a).
By taking a $2 \times 2$ unit cell (which is twice as large
as the minimal magnetic unit cell), the dispersion relation is
$E_{\pm}=\pm \sqrt{4+2\cos{2 k_x}-2\cos{2 k_y}},$
where $(k_x,k_y)$ is the wavenumber which belongs to the
reduced Brillouin zone $-\pi/2 \leq k_{x,y} < \pi/2$. The bands in the first Brillouin zone are shown in Fig.~\ref{fig.pi_flux} (b).
Each energy level is doubly degenerate.
The ground-state energy of fermions at zero chemical potential, which
corresponds to the half filling, is given as
$E_0^{\rm F} = \sum_{k_x,k_y} 2 E_{-}(k_x,k_y)$ ,
where the factor $2$ comes from the double degeneracy.
For the square lattice of size $L_x \times L_y$ ($N=L_x L_y$),
$k_{x,y}$ is respectively quantized as integral multiples of
$2\pi/L_{x,y}$.
Thus, in the thermodynamic limit $L_{x,y} \to \infty$, the ground-state energy of the fermionic model at $\mu=0$ is obtained exactly as
\begin{eqnarray}
\frac{E_0^{\rm F}}{N}& = &- \frac{1}{2} \int_{-\pi}^\pi \frac{d \tilde{k}_x}{2\pi} \;
\int_{-\pi}^\pi \frac{d \tilde{k}_y}{2\pi} \;
\sqrt{4+2\cos{\tilde{k}_x}-2\cos{\tilde{k}_y}}\nonumber\\
 &=& - 0.958091.
 \label{eq.exact}
\end{eqnarray}
The extrapolated ground-state energy density of fermions from finite-size scaling in Fig.~\ref{fig.scaling_half_filling} matches well with the exact result.

\begin{figure}[htbp]
\centering \includegraphics[width=3.0in]{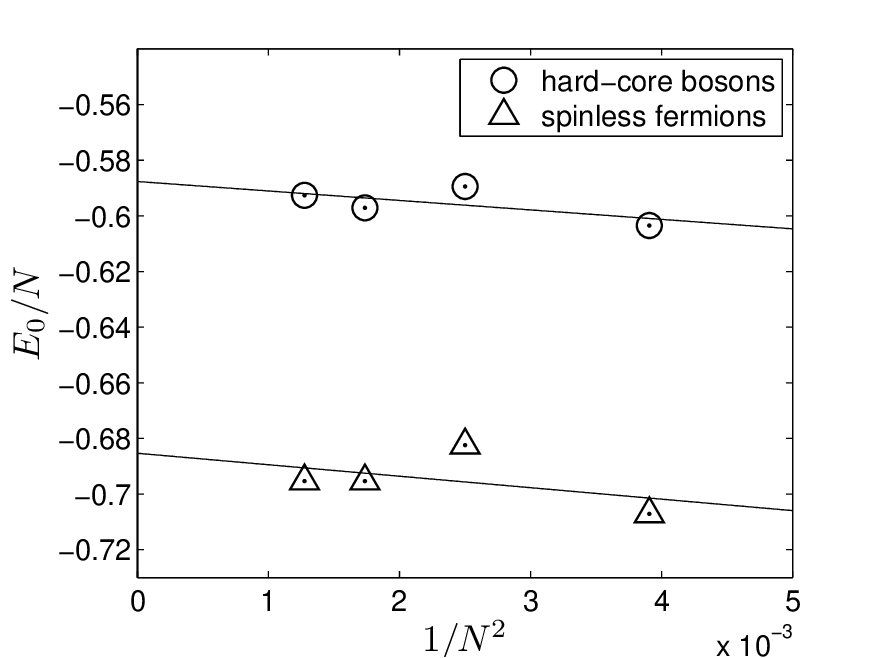}
\caption{Finite-size scaling of ground-state energies in two-dimensional
square lattice with $\Phi_0/4$ flux per plaquette at quarter filling.
The fitting functions are $E_0^{\rm B}/N= -0.5877-3.405/N^2+O(N^{-4})$ for
hard-core bosons and $E_0^{\rm F}/N =-0.6853-4.125/N^2+O(N^{-4})$ for
fermions, respectively.}  \label{fig.scaling_quarter_filling}
\end{figure}

We consider the grand canonical ground-state energy of bosons at the same chemical potential ($\mu=0$). We rewrite the Hamiltonian
$\calH =\sum_\alpha h_\alpha$, where
$h_{\alpha}=-\frac{1}{2}\sum_{\langle j,k\rangle\in\boldsymbol{+}_\alpha}
( t_{jk}c_j^{\dag}c_k+\textrm{H.c.})$ is the cluster Hamiltonian defined on a $12$-site cross-shaped cluster as shown in Fig.~\ref{fig.pi_flux} (a).
The whole lattice is covered by the brown cross-shaped clusters with the same pattern of hopping amplitudes within the cluster, whose centers are denoted by the black dots.
Therefore, each cluster
overlaps with $4$ neighboring clusters and
each link appears in two different clusters
when periodic boundary conditions are imposed.
The factor $1/2$ in $h_{\alpha}$ compensates this double counting.
By the Anderson's argument~\cite{Anderson,Tarrach-lower-bound,Valenti-lower-bound}, the ground-state energy $E_0^{\rm B}$ of $\calH^{\rm B}$ satisfies
$E_0^{\rm B}\geq\sum_\alpha\epsilon_0^\alpha$,
where $\epsilon_0^\alpha$ is the ground-state energy of $h_\alpha$.
The ground-state energy of $h_{\alpha}$ on a cluster with a given particle number $m$  obtained by exact diagonalization is shown in Table~\ref{table.pi_flux} in Appendix~\ref{app:tables}. The grand canonical ground-state energy of the cross-shaped cluster is obtained
as $\epsilon_0^\alpha = - 3.609035$.
Assuming the number of sites in the square lattice is $N$,
we obtain
\begin{equation}
E_0^{\rm B}/N\geq-3.609035/4=-0.902259>E_0^{\rm F}/N,
\end{equation}
where $N/4$ is the number of clusters.
Thus hard-core bosons have a higher ground-state energy than fermions at half filling ($\mu=0$), even in the thermodynamic limit,
as expected from extrapolation from finite-size scaling and statistical transmutation argument~\cite{ourPRL,Semenoff,Fradkin,sedrakyan-PRA,sedrakyan-PRB}.

We note that the choice of cluster decomposition is not unique for a given model.
In order to prove the reversal of the natural inequality,
an appropriate choice of the cluster decomposition with a sufficiently high lower bound for the ground-state energy of bosons relative to that of fermions is necessary.
Here we have discussed the decomposition into cross-shaped clusters,
which can be handled relatively easily but is still useful for proving the reversed natural inequality.
Decomposition into larger clusters is expected to give a more precise estimation of a lower bound. 
Similar comment also applies to the cluster decompositions discussed in Sec.~\ref{sec:flat-band}.

For other values of flux per plaquette or filling fraction,
there is no rigorous proof available at present.
However, the finite-size scaling of numerical data with $\Phi_0/4$
flux per plaquette at quarter filling, shown in
Fig.~\ref{fig.scaling_quarter_filling},
suggests that fermions have a lower
ground-state energy in the thermodynamic limit.

The violation of the natural inequality in systems with flux is not restricted to two dimensions. We have indeed proved that the natural inequality could be reversed in a tight-binding model on a three-dimensional pyrochlore lattice with flux~\cite{ourPRL}.

\subsection{Cluster decomposition in flat band models}
\label{sec:flat-band}

In this section, we present a rigorous proof that
the reversed natural inequality
also holds in several flat-band models,
even in the thermodynamic limit.
Although the
existence of a flat band is neither a necessary nor sufficient condition
to violate Eq.~\eqref{eq.e0b.leq.e0f}, it does tend to help: when the
lowest flat band is occupied by the fermions, there is no extra energy
gain due to Pauli exclusion principle. Therefore, the inversion of the
natural inequality has a better chance to be realized in flat band models.
Here we show that the inequality~\eqref{eq.e0b.leq.e0f} is indeed
violated in a few examples with flat bands, by
a cluster decomposition technique.

First we discuss the delta-chain model, for which the violation of
Eq.~\eqref{eq.e0b.leq.e0f} was numerically found for small
clusters~\cite{HuberAltman2010,Altman-priv}.  The Hamiltonian of the
model can be written in the following
form~\cite{Tasaki-flatband,Mielke-flatband}:
\begin{equation}
\mathcal H=\sum_{j=1}^{N}a_j^{\dag}a_j, \label{eq.deltachain}
\end{equation}
where the $a$-operator, which acts on each triangle, is defined as
$a_j=c_{2j-1}+\sqrt{2}c_{2j}+c_{2j+1}$.  Periodic boundary condition is
used to identify $c_{2N+1}$ with $c_1$.  The Hamiltonian $\mathcal H$
corresponds to a model with negative hopping amplitudes $t_{jk}$ (as
defined in Eq.~\eqref{eq.Ham}), which lead to frustration.

\begin{figure}
\includegraphics[width=0.8\columnwidth]{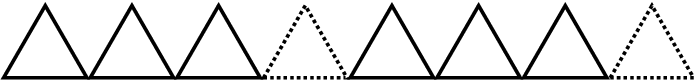} \caption{
An example of decomposition of the delta-chain Hamiltonian to clusters,
with $p=4$ unit cells
including one decoupled
site at the top of the dashed triangle.}
\label{fig.deltachain.p4}
\end{figure}

The model in the single-particle sector has two bands.  The lower flat
band with zero energy is spanned by states annihilated by $a_j$'s.  We
note that the Hamiltonian~\eqref{eq.deltachain} is modified from that in
Ref.~\onlinecite{HuberAltman2010} by a constant chemical potential, so that
the flat band has exactly zero energy.
Thus the ground-state energy of
the fermionic version of the model~\eqref{eq.deltachain} is zero as long
as the filling fraction $\nu$ satisfies $\nu\leq1/2$.

On the other hand, in general, construction of the ground state
of a system of many interacting bosons is not straightforward
even if the single-particle states are known exactly.
However, the flat band in the geometrically
frustrated antiferromagnet also implies the existence of
non-overlapping localized zero-energy states.
It was first pointed out in Ref.~\onlinecite{Schulenburg-PRL2002},
and was later applied to various
problems\cite{Schnack-JPCM2004,*Derzhko-Richter2004,
*Zhitomirsky-Tsunetsugu,*Zhitomirsky-Tsunetsugu-long,*Derzhko2007summary,Zhitomirsky-Honecker2004}.
In the case of the delta chain,
the ground-state energy $E_0^{\rm B}$ of bosons is zero as long as
$\nu\leq1/4$, since each boson can occupy different non-overlapping
localized zero-energy state\cite{Schulenburg-PRL2002,Zhitomirsky-Honecker2004,Schnack-JPCM2004}.

Now let us derive a nontrivial lower bound for $E_0^{\rm B}$ for filling
fractions $\nu >1/4$.  We decompose the model into clusters, each
containing $p$ unit cells:
\begin{equation}
\mathcal H= \sum_{n=0}^{N/p-1}\mathcal
H_{n}^{(p)}+\sum_{n=1}^{N/p}a_{np}^{\dag}a_{np},
\end{equation}
where $\calH_{n}^{(p)}=\sum_{j=1}^{p-1}a^{\dagger}_{np+j}a_{np+j}$ is
the Hamiltonian for the solid triangles as in
Fig~\ref{fig.deltachain.p4}.  Since the second term
$\sum_{n=1}^{N/p}a_{np}^{\dag}a_{np}$,
describing hoppings on dashed triangles, is positive semidefinite, the
ground-state energy $\tilde{E}_0^{\rm B}$ of the first term
$\tilde{\calH}=\sum_{n=0}^{N/p-1}\mathcal H_{n}^{(p)}$ satisfies
$\tilde{E}_0^{\rm B}\leq E_0^{\rm B}$.  $\tilde{\calH}$ is a sum of
mutually commuting cluster Hamiltonians $\calH_{n}^{(p)}$.  Thus
$\tilde{E}_0^{\rm B}$ is simply given by the sum of the ground-state
energies of all clusters.
The particle number within each cluster is also conserved separately in
$\tilde{\calH}$.  Let us choose $p=4$ as in
Fig.~\ref{fig.deltachain.p4}, so that the cluster contains 8 sites.
The ground-state energy in each sector with fixed particle number $m$ is
obtained by 
exact diagonalization of the $8$-site cluster, which is shown in Table~\ref{table.delta-chain} in Appendix~\ref{app:tables}.
We find
$\epsilon_0^{(4)}(m)\geq\Delta_{\mbox{\scriptsize DC}}^{(4)}=0.372605$
for $4\leq m\leq 8$,
while $\epsilon_0^{(4)}(m)=0$ for $0\leq m\leq3$.

If we consider the filling fraction in the range $3/8<\nu\leq1/2$, it
follows from Dirichlet's box principle that there is at least one
cluster which contains 4 or more particles.  Thus, in this range,
$\tilde{E}_0^{\rm B}\geq\Delta_{\mbox{\scriptsize DC}}^{(4)}$ for any
system size $N$, while $E_0^{\rm F}=0$.  Therefore, the inversion of the
ground-state energies holds also in the thermodynamic limit.

The outcome of the above argument depends on the cluster size taken.  In
fact, the range of filling fraction $\nu$ for which we have proved the
violation of Eq.~\eqref{eq.e0b.leq.e0f} is not optimal.
In Appendix.~\ref{app:optimal_bound}, using a different
technique, we will extend the range to $1/4 < \nu \le 1/2$;
the lower bound $1/4$ is in fact optimal.

This method can be easily extended to other lattices.  For example, the
standard nearest-neighbor hopping model on the kagome lattice can be written
as
\begin{equation}
\mathcal H = \sum_{\alpha}
a_{\bigtriangleup_\alpha}^{\dag}a_{\bigtriangleup_\alpha} +\sum_{\alpha}
a_{\bigtriangledown_\alpha}^{\dag} a_{\bigtriangledown_\alpha},
\label{eq:kagome}
\end{equation}
where $\bigtriangleup_\alpha$ and $\bigtriangledown_\alpha$ are
elementary triangles pointing up and down, respectively, of the kagome
lattice, as shown in Fig.~\ref{fig.star.david}.  We define
$a_{\bigtriangleup_\alpha} \equiv c_{\alpha_1} + c_{\alpha_2} +
c_{\alpha_3}$, where $\alpha_{1,2,3}$ refer to the three sites belonging
to $\bigtriangleup_\alpha$, and likewise for
$a_{\bigtriangledown_\alpha}$.  The fermionic version of the model has
three bands, the lowest of which is a flat band at zero
energy~\cite{Mielke-kagome,Mielke-flatband,Bergman}.
Thus $E_0^{\rm F} =0$ when $\nu \leq 1/3$.

For the ground-state energy of the bosonic version, we can use the
cluster decomposition technique similar to what we have discussed above
for the delta-chain.  Let us choose the 12-site cluster of the ``Star of
David'' shape, which is shown by solid lines in Fig.~\ref{fig.star.david}.  The ground-state energy of the cluster in each sector
with $m$ particles is shown in Table~\ref{table.kagome} in Appendix~\ref{app:tables}.  The
ground-state energy $\epsilon_0^{\textrm{cluster}}$ of each cluster is
zero with $m\leq 3$, but is positive with $m \geq 4$.  Thus, invoking
Dirichlet's box principle again, Eq.~\eqref{eq.e0b.leq.e0f} is violated
for filling fraction $1/4 < \nu \leq 1/3$.
This conclusion also holds in the thermodynamic limit,
where the system size $N$ is taken to the infinity while keeping
the filling fraction $\nu$ constant.

\begin{figure}
\includegraphics[width=0.4\columnwidth]{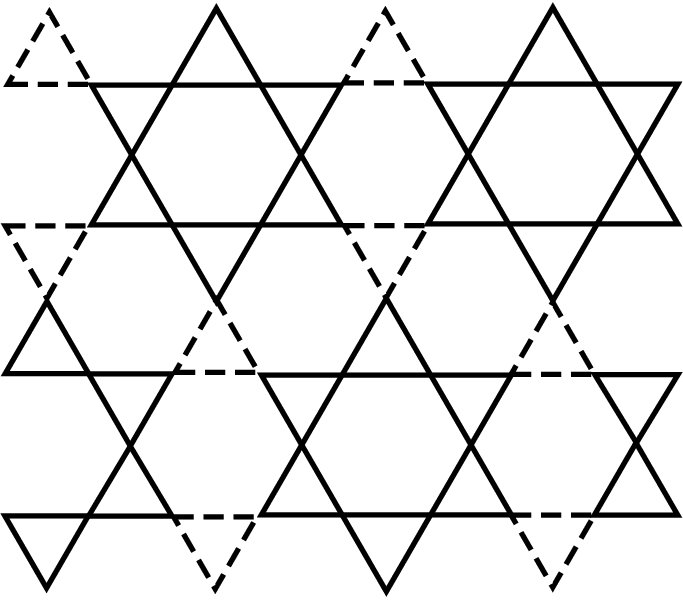} \caption{The
$12$-site clusters of ``Star of David'' shape are shown in solid lines
on kagome lattice.}  \label{fig.star.david}
\end{figure}

\subsection{\label{interaction}Extension to interacting systems}

Throughout most of this paper, we limited the interactions to the
hard-core ones for technical simplicity: fermions are then free, while bosons
are subject only to the hard-core interaction.
Here we comment briefly on the effect of the other possible interactions.
Theorems~\ref{thm.natural_spinless},~\ref{thm.natural_spinful} and~\ref{thm.natural_infiniteU}
are actually valid even in the presence of density-density interactions other than the hard-core interaction.
Introduction of additional density-density interactions
should not essentially modify the comparison of the ground-state energies, as it would affect bosonic
and fermionic models in a similar manner. For example, the interaction terms are introduced in diagonal terms in the matrix of Hamiltonian in Theorem~\ref{thm.natural_spinless}, which do not affect the conclusion of the comparison. Therefore, in order to understand the essence of physics in the present
problem, it would suffice to consider the hard-core interactions only.

That said, in fact, one can actually prove that
the inequality~\eqref{eq.e0b.leq.e0f}
is violated even in the presence of an additional
density-density interactions in the one-dimensional ring with $\pi$ flux
discussed in Sec.~\ref{sec:ring}.
This can be seen by noting that Jordan-Wigner transformation applies
regardless of the presence of density-density interactions,
and implies
\begin{equation}
E_0^{\rm F}(\Phi=\pi)=E_0^{\rm B}(\Phi=0), 
\end{equation}
where the number of particles is assumed to be even.
Then we see that a lattice version of Simon's
theorem~\cite{Barry_Simon_diamag}
also applies in the presence of the interaction:
\begin{equation}
E_0^{\rm B}(\Phi=\pi) \geq E_0^{\rm B}(\Phi=0),
\end{equation}
giving $E_0^{\rm B}(\Phi=\pi) \geq E_0^{\rm F}(\Phi=\pi)$.
Furthermore, under appropriate assumptions, it is possible
to prove the strict inequality
$E_0^{\rm B}(\Phi=\pi) > E_0^{\rm F}(\Phi=\pi)$
in the presence of interaction, with an argument similar
to the proof of Theorems~\ref{thm.natural_spinless}
and~\ref{thm.diamagnetic}.

\section{\label{conclusions}Conclusions and discussions}

In this paper, we have proved that the
ground-state energy of hard-core bosons is lower than that of fermions if
there is no frustration in the hopping.

The effect of the statistical phase of fermions can then be understood as a
frustration, since it results in destructive quantum interferences among
different paths.
In fact, the phase introduced by Fermi
statistics can be effectively described by a magnetic flux, after the
mapping to the single-particle tight-binding model on a fictitious
lattice which represents the Fock space.
In this sense, the non-strict version of the natural inequality
is a corollary of the lattice version of the
diamagnetic inequality.
On the other hand, we also proved a strict version of the
natural inequality, under certain conditions.
The key of the proof is the contribution of an exchange
process of two particles, which is exactly what demonstrates
the statistics of the particles.
The argument is also applied to prove the
strict version of the diamagnetic inequality on the lattice.

Once a magnetic flux is introduced in the original many-particle
problem, the hopping terms can be frustrated.  The hopping frustration
can partially cancel the statistical frustration of fermions, hinting at
the possibility that the natural inequality can be reversed in the
presence of hopping frustration.  We proved rigorously that the natural
inequality is indeed reversed in the presence of frustration, in various
examples. They include one-dimensional $\pi$-flux ring, coupled rings in
two dimensions, 
systems with flux in $2$D and $3$D, flat band
models by cluster decomposition technique. 
Finally, we demonstrated an example of the
violation of natural inequality with other interaction than hard-core
constraint.

In this paper, we focused on the case of hard-core bosons for
simplicity.  However,
Theorems~\ref{thm.natural_spinless},~\ref{thm.natural_spinful}
and~\ref{thm.natural_infiniteU} can be readily generalized to soft-core
bosons.  This is because hard-core bosons can be regarded as a special
limit of more general interacting bosons.  That is, we can introduce the
on-site interaction $\frac{\calU}{2} n_i (n_i-1)$; the hard-core
constraint can be then implemented by taking $\calU\to+\infty$. The
on-site interaction term is positive semi-definite for bosons, if $\calU
\geq 0$.  Thus the hard-core bosons have a higher ground-state energy
than that of soft-core bosons at finite $\calU$.  This implies the
applicability of Theorems~\ref{thm.natural_spinless},
\ref{thm.natural_spinful} and~\ref{thm.natural_infiniteU} to the
soft-core bosons.

Our analysis of the hard-core boson model also suggests that
the natural inequality for soft-core bosons could be reversed by
introducing the hopping frustrations.
However, soft-core bosons are closer to free bosons,
which never violate the natural inequality because of
the simple argument based on perfect BEC.
Thus the violation would be more difficult to be realized
in soft-core bosons, compared to the hard-core bosons
discussed in this paper.
Other open problems include comparison in the presence of
other degrees of freedom such as the orbital/flavor
of particles.
The non-strict version of the theorems
can be easily generalized to the case with multiple orbitals/flavors.

In this paper, we have also discussed briefly the comparison of
the ground-state energies of spinful bosons and fermions.
The natural inequality still holds in the absence of hopping
frustration.
Although we did not discuss explicitly for spinful particles,
the natural inequality is expected to be violated
by introducing appropriate hopping frustration.

Here it should be recalled that, physical magnetic field not
only introduces phase factors in hopping terms, but is also
coupled to the spin degrees of freedom via Zeeman term.
Thus, Zeeman term should be also taken into account,
in order to discuss a physical magnetic field applied to
the system of charged particles.
The Zeeman term acts as different chemical potentials for
up-spin and down-spin particles.
Thus much of the discussion in the present paper is still applicable.
For example, in the absence of hopping frustration,
the natural inequality still holds even in the presence of
the Zeeman term.
Once hopping frustration is introduced, the natural inequality
can be violated.
However, exactly how the violation of the natural inequality occurs
does depend on the chemical potential, and on the
Zeeman effect in the case of spinful particles.

On the other hand, we also note that phase factors in hopping terms and
Zeeman coupling are two distinct effects, which in principle can be
controlled independently.  In fact, for neutral cold atoms, the phase
factor in hoppings are usually introduced as ``synthetic gauge
field''~\cite{rev-gauge-field}, instead of the physical magnetic
field. This does not produce Zeeman coupling, making it possible to
study the effect of hopping frustrations separately from that of the
Zeeman effect.

\section{Acknowledgement}

We are grateful to
Ehud Altman, Claudio Chamon, Sebastian Huber,
Fumihiko Nakano, Xiwen Guan, Naoki Kawashima, Naomichi Hatano, Hui-Hai Zhao, Zheng-Yu Weng and Long Zhang for the valuable discussions and comments.
W.-X. N. is supported by NSFC
(11704267) and start-up funding from Sichuan University (2018SCU12063), and
MEXT scholarship during the early stage of this work.
M. O. was supported in part by
Grants-in-Aid for Scientific Research (KAKENHI) Nos. JP25103706 and JP16K05469.
H.K. was supported in part by JSPS KAKENHI Grant No. JP23740298 and JP15K17719.
A part of the present work was carried out
during a visit of W.-X. N. and M. O. to Kavli Institute
for Theoretical Physics, UC Santa Barbara, supported by
US National Science Foundation Grant No. NSF PHY11-25915.
Part of the numerical
calculation is carried out by
TITPACK ver.2, developed by H. Nishimori.

\appendix

\section{Proof of Theorem $2$}
\label{app:proof2}

\begin{proof}
 Since the total number operator $M=\sum_{j\sigma}n_{j\sigma}$
 and total magnetization
 $S_z=1/2\sum_{j}(n_{j\uparrow}-n_{j\downarrow})$
 commute with the Hamiltonian~\eqref{eq.spinfulHam}, one can
 diagonalize the Hamiltonian in each
 sub-Hilbert space with fixed values of $M$ and $S_z$.
 Each sub-Hilbert space has definite numbers of up-spin and down-spin
 particles.
 Let
 $|\phi^{\mu}\rangle_{\uparrow}\equiv|\{n_{j\uparrow}^{\mu}\}\rangle$
 ($\mu=1,2,\cdots,u$) be the occupation number basis for up-spin
 particles,  and
 $|\psi^{\nu}\rangle_{\downarrow}\equiv|\{n_{j\downarrow}^{\nu}\}\rangle$
 ($\nu=1,2,\cdots,v$) be the occupation number basis for down-spin particles.
 Then, we can take the direct product
 $|\Phi^{a}\rangle=|\psi^{\nu}\rangle_{\downarrow}\otimes
 |\phi^{\mu}\rangle_{\uparrow}$,
 where $a=1,2,\cdots, uv$, as the basis of the sub-Hilbert
 space mentioned above.

 The Hamiltonian can be rewritten as:
 \begin{eqnarray}
 \calH&=&\calH_{\textrm{t}}+\calH_{\textrm{int}},\\
 \calH_{\textrm{t}}&=&\mathds{1}^{\downarrow}\otimes\calH_{\textrm{t}}^{\uparrow}
 +\calH_{\textrm{t}}^{\downarrow}\otimes\mathds{1}^{\uparrow},
 \end{eqnarray}
 where $\calH_{\textrm{t}}^{\sigma}=-\sum_{j\ne
 k}(t_{jk}c_{j\sigma}^{\dag}c_{k\sigma}+\mbox{H.c.})$.
The matrix elements of
 the number operator $n_{j\sigma}$ are the same in this basis, for
 hard-core bosons and fermions. We introduce the operator $\calK^{\rm
 B,F} \equiv - \calH^{\rm B,F} + C \mathds{1}$ with a constant $C$.
 Choosing $C$ large enough, we make all the eigenvalues and all the diagonal
 matrix elements  of $\calK^{\rm B,F}$ positive.
 The matrix elements of bosonic and
 fermionic Hamiltonians obey the relation:
 \begin{equation}
 \calK_{ab}^{\rm B}=\left\{
 \begin{array}{ll}
 |\calK_{ab}^{\rm F}|&(a\ne b)\\ \calK_{aa}^{\rm F}& (a=b),
 \end{array}
 \right.
 \end{equation}
 where the diagonal terms correspond to $\calH_{\textrm{int}}$ and the
 off-diagonal terms correspond to $\calH_{\textrm{t}}$. 
The non-strict inequality is easily proved by variational principle in the same manner employed in Proof of Theorem~\ref{thm.natural_spinless}. Here, we focus on the discussion on strict natural inequality for spinful case with finite $U_j$.
 
 With finite $U_j$'s, one site
 can be occupied by one spin-up particle and one spin-down particle.
 Thus spin-up particles can move as spinless particles for any given
 configuration of spin-down particles, and vice versa.
 Of course, the interaction term $\calH_{\textrm{int}}$, which is diagonal
 in this basis, is affected by the presence of particles with opposite
 spins. However, as far as the irreducibility (connectivity) of
 Hamiltonian is concerned, one can regard the system as
 two independent systems of hard-core particles.
 As a consequence, when the lattice $\Lambda$ is connected, any
 pair of basis states $|\Phi^a\rangle_{\rm B}$ and
 $|\Phi^{b} \rangle_{\rm B}$  are connected to each other
 by successive applications of the hopping term in $\calK^{\rm B}$.
 Together with the property
 $\calK_{ab}^{\rm B}\ge 0$, $\calK^{\rm B}$ satisfies the condition of the
 Perron-Frobenius theorem. When the number of particles $M\ge3$, there
 are at least two particles with the same spin. The condition $M\le2N-3$
 guarantees that there are at least two spaces which can accommodate two
 particles with the same spin. Thus, when the number of particles falls
 in the range $3\le M\le2N-3$, we can exchange two identical particles
 and return back to the same state, based on the branch structure as in
 Fig.~\ref{fig.particle-exchange}.
 Therefore, when $U_j$'s are finite, the lattice
 is connected and has a branch structure, and $3\le M\le 2N-3$,
 two-particle exchange always happens.
 As in the proof of Theorem~\ref{thm.natural_spinless}
 for spinless case, the strict inequality $E_0^{\rm B} < E_0^{\rm F}$
 follows from the Perron-Frobenius theorem.
 \end{proof}

\section{Remarks about the proof of Theorem $3$ and discussions}
\label{app:proof3}
The no-strict inequality remains unaffected by taking $U_j=+\infty$. Here, we focus on a sufficient condition for strict natural inequality with infinite repulsion, for spinful case.

With infinite on-site repulsion, the maximum number of
particles is $N$. The condition $M\ge 3$ is to guarantee there are at
least two particles with the same spin such that they can be
exchanged. For a lattice connected by exchange bonds, two particles on
an exchange bond can be exchanged without changing the configuration
outside, by hopping a hole around
the loop on which both the exchange bond and the hole lie~\cite{tasaki-1998}.
Hence, when the number of particle $M$
satisfies $3\le M\le N-1$, two particles with the same spin can be exchanged on an
exchange-bond lattice by successive particle hoppings.

The property that the entire lattice is connected by exchange bonds
can be verified~\cite{tasaki-1998} in various common lattices,
such as triangular, square, simple cubic, fcc, or bcc
lattices, in which nearest neighbor sites are connected
by nonvanishing hopping amplitudes.
Thus, the above theorem holds for these lattices.

We also note that, Nagaoka's ferromagnetism only applies to
the system with single hole with respect to half filling.
However, this restriction is only necessary
to guarantee that all the matrix elements are nonnegative.
The irreducibility of the Hamiltonian matrix does not
require that there is only one hole.
In fact, the breakdown of the positivity in the presence of
more than one holes in the Hubbard model with
$U_j = +\infty$ is precisely due to the Fermi statistics
of the electrons.
If we consider the ``Bose-Hubbard model'' with spin-$1/2$ bosons
instead of electrons, all the matrix elements are nonnegative
in the occupation number basis, for any number of holes.
Thus the Bose-Hubbard model with spin-$1/2$ bosons exhibit
ferromagnetism for any
filling fraction~\cite{EisenbergLieb}.
This nonnegativity of the matrix elements for bosons
is also essential for Theorem~\ref{thm.natural_infiniteU},
which holds for any filling fraction.

The proofs of Theorems~\ref{thm.natural_spinless}
and~\ref{thm.natural_spinful} are insensitive to the signs of the
interaction terms $V_{jk}$ and $U_j$. Namely the natural inequality holds no
matter the interaction is repulsive or attractive. The interesting
aspect of the attractive interaction is that it will induce Cooper pair
of fermions.  In the case of spinless fermions, orbital part of the
Cooper pair wavefunction must be antisymmetric with respect to
the exchange of two fermions. This results in an extra cost in
the kinetic energy. Such a fermionic BEC state thus has a higher ground-state
energy than its bosonic counterpart, in full agreement of Theorem~\ref{thm.natural_spinless}.

In contrast, in the case of spinful fermions, with attractive
interaction, fermions could pair up in the nodeless $s$-channel.
In this case, there is no obvious reason why the fermions have
a higher ground-state energy than bosons.
Nevertheless, according to Theorem~\ref{thm.natural_spinful},
spinful fermions still have strictly higher ground-state energy
than corresponding bosons, even when the pairing is in the
nodeless $s$-channel.

This can be interpreted physically in the
following way. If the paring of two particles is completely robust, the
problem is reduced to the identical problem of bosonic ``molecules'',
whether the original particles are fermions or bosons. Then the ground-state energies should be the same for fermions and bosons. However, in
general, the pairing is not completely robust, and two pairs can
(virtually) exchange each one of their constituent particles. The
amplitude for such a process has negative sign only for fermions,
leading to the nonvanishing energy difference between fermions and
bosons. The exception occurs when the on-site attractive interaction
between up and down spin particles is infinite ($U_j=-\infty$). Then the
pairs are completely robust, and no virtual exchange of constituent
particles occurs; the ground-state energies for fermions and bosons
become identical in this limit. On the other hand, with the infinite
attraction, the irreducibility can not be satisfied.
Because a hopping of a molecule requires its breaking,
which costs an infinite energy and is thus prohibited.
This implies that the bosonic molecules are completely
localized in the model~\eqref{eq.spinfulHam}.
Thus the natural inequality is reduced to the trivial
equality $E_0^{\rm B}=E_0^{\rm F}$ in the limit $U_j\to -\infty$ .

\begin{figure}[htbp]
\centering \includegraphics[scale=0.6]{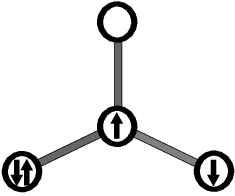} \caption{A $4
$-site lattice with four spins at half filling and $S_z=0$.}
\label{fig.branch}
\end{figure}
\begin{figure}[htbp]
\centering \includegraphics[scale=0.7]{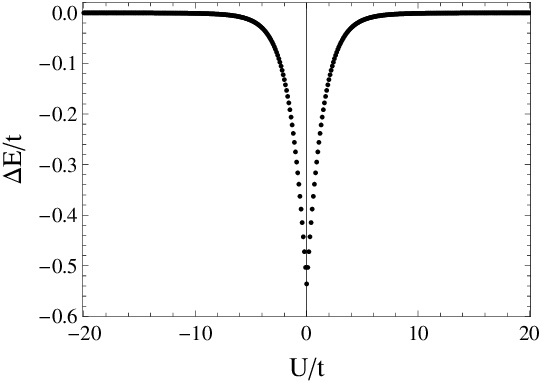} \caption{Difference
of ground-state energy ($\Delta E=E_0^{\rm B}-E_0^{\rm F}$) between
hard-core bosons and fermions on the $4$-site lattice with a branch, in
$S_z=0$ sector with $4$ spins. The absolute value of energy difference
decreases down to $\sim 10^{-8}t$ around $|U|/t=100$.}
\label{fig:Hubbard}
\end{figure}

In the following, as an example, we numerically demonstrate above observations in
spinful hard-core Bose-Hubbard and Fermi-Hubbard models
on a 4-site lattice as shown in Fig.~\ref{fig.branch}.
The Hamiltonian is given by
\begin{equation}
\mathcal H=-t\sum_{\langle
i,j\rangle}\sum_{\sigma}(c_{i\sigma}^{\dag}c_{j\sigma}+\textrm{H.c.})  +U\sum_j
n_{j\uparrow}n_{j\downarrow},
\end{equation}
where $t>0$, $\langle i,j\rangle$ denotes a pair of neighboring sites,
and the hard-core constraint $n_{j\sigma}=0,1$ is again imposed for the bosons.
We consider the spin-$1/2$ bosons and fermions at half filling
(the total number of particles per site $\nu = 1$) and $S^z=0$.
That is, on this 4-site cluster, there are two up-spin particles
and two down-spin particles.
The energy difference between spinful bosons and spinful fermions ($\Delta
E=E_0^{\rm B}-E_0^{\rm F}$) is shown as a function of $U=U_j$ in
Fig.~\ref{fig:Hubbard}.

Conforming to Theorem~\ref{thm.natural_spinful},
$E_0^{\rm B}\le E_0^{\rm F}$ holds for all range of $U$, independent of
the sign of $U$. Moreover, $\Delta E (U)$ is symmetric along $U=0$ due
to particle-hole symmetry of Hubbard model at half filling
$\nu=1$~\cite{PHS-Hubbard}.

When $U$ is finite, fermions have strictly higher
ground-state energy than bosons, again in agreement with
the latter half of Theorem~\ref{thm.natural_spinful}.
When $U=+\infty$, on the other hand,
the particles are completely immobile at half-filling
and thus no particle-exchange occurs.
The ground-state energy is indeed exactly the same for
fermions and for bosons in this limit.
Likewise, in the limit of $U=-\infty$, either bosons or fermions
form completely robust (and immobile) pairs, and the ground-state
energies are exactly the same. In the present case, this can
also be understood as a consequence of the particle-hole
symmetry at half filling~\cite{PHS-Hubbard},
which maps $U \rightarrow -U$.

\section{Proof of Theorem $4$}
\label{app:proof4}

\begin{proof}
The proof is similar to that of Theorem~\ref{thm.natural_spinless}.
We can define the matrices $\calK$, $\calK'$ by
\begin{align}
\calK_{\alpha \beta} & \equiv \tau_{\alpha\beta} + C \delta_{\alpha\beta},
\\
\calK'_{\alpha \beta} & \equiv \tau'_{\alpha\beta} + C \delta_{\alpha\beta},
\end{align}
with a sufficiently large constant $C$ so that $\calK$ and $\calK'$
is positive definite.
We then define $\calL \equiv \calK^n$ and $\calL' \equiv {\calK'}^n$,
for the length $n$ of the loop with a nonvanishing flux.
The positive definiteness of $\calK$ and $\calK'$ implies
that $\calL$ and $\calL'$ are also positive definite,
and thus all the diagonal matrix elements $\calL_{\alpha \alpha}$
and $\calL'_{\alpha \alpha}$ are strictly positive.
Similarly to the proof of Theorem~\ref{thm.natural_spinless},
$\calL'_{\alpha\beta}\ge|\calL_{\alpha\beta}|$ holds for any
$\alpha,\beta$.
In particular, the diagonal matrix elements of $\calL'$ and $\calL$ are expanded as
\begin{align}
\calL'_{\alpha_0\alpha_0} &=
\sum_{\alpha_1,\cdots,\alpha_{n-1}}
\calK'_{\alpha_0 \alpha_1} \calK'_{\alpha_1 \alpha_2}
\ldots \calK'_{\alpha_{n-1} \alpha_0},
\\
\calL_{\alpha_0\alpha_0} &=
\sum_{\alpha_1,\cdots,\alpha_{n-1}}
\calK_{\alpha_0 \alpha_1} \calK_{\alpha_1 \alpha_2}
\ldots \calK_{\alpha_{n-1} \alpha_0}.
\end{align}
Each term in the expansion satisfies
\begin{align}
\calK'_{\alpha_0 \alpha_1} \calK'_{\alpha_1 \alpha_2}
\ldots \calK'_{\alpha_{n-1} \alpha_0}
\ge
\left| \calK_{\alpha_0 \alpha_1} \calK_{\alpha_1 \alpha_2}
\ldots \calK_{\alpha_{n-1} \alpha_0} \right| ,
\end{align}
thanks to
$\calK'_{\alpha\beta} \ge | \calK_{\alpha\beta}|$.
By assumption, there is a nonvanishing
contribution to $\calL_{\alpha_0 \alpha_0}$
from the loop of length $n$,
\begin{align}
\calK_{\alpha_0 \alpha_1} \calK_{\alpha_1 \alpha_2}
\ldots \calK_{\alpha_{n-1} \alpha_0}
=
\tau_{\alpha_0 \alpha_1} \tau_{\alpha_1 \alpha_2}
\ldots \tau_{\alpha_{n-1} \alpha_0},
\label{eq.loop.contrib}
\end{align}
which is not positive.
Here we used the fact that the off-diagonal elements of
$\calK$ and $\tau$ are identical.
Combining with the contribution from its reverse loop
\begin{align}
\calK_{\alpha_0 \alpha_{n-1}} \calK_{\alpha_{n-1} \alpha_{n-2}}
\ldots \calK_{\alpha_1 \alpha_0},
\end{align}
which is the complex conjugate of Eq.~\eqref{eq.loop.contrib},
we find the strict inequality
\begin{multline}
\calK'_{\alpha_0 \alpha_1} \calK'_{\alpha_1 \alpha_2}
\ldots \calK'_{\alpha_{n-1} \alpha_0} + \mbox{c.c.}
\\
>
\calK_{\alpha_0 \alpha_1} \calK_{\alpha_1 \alpha_2}
\ldots \calK_{\alpha_{n-1} \alpha_0} + \mbox{c.c.} .
\end{multline}
Thus $\calL'_{\alpha_0\alpha_0} > \calL_{\alpha_0\alpha_0} > 0$.
Invoking the Perron-Frobenius theorem again, the strict diamagnetic
inequality~\eqref{eq.strict_dia} is proved.
\end{proof}

\section{supporting results of diagonalization involved in this work}
\label{app:tables}
The results of numerical exact diagonalization on finite lattices are presented here to assist the proofs in the main text.

\begin{table}
\begin{center}
\begin{tabular}{|c|c|c|}
\hline $m$ & $\epsilon_{\Diamond}^{\rm F}(m)$ &
$\epsilon_{\Diamond}^{\rm B}(m)$ \\ \hline \hline 1 &
$-\sqrt{2}t+t'$ & $-\sqrt{2}t-t'$ \\ \hline 2 & $-2\sqrt{2}t+2t'$ &
$-2t-2t'$\\ \hline 3 & $-\sqrt{2}t+3t'$ & $-\sqrt{2}t-3t'$\\ \hline 4 &
$4t'$ & $-4t'$\\ \hline
\end{tabular}
\end{center}
\caption{The ground-state energies of fermions and hard-core bosons on a
thick-line square as shown in Fig.~\ref{fig.octagon_square} (a), where $m$ is the number of particles on a
$\pi$-flux square. The results are used in the proof of $\pi$-flux octagon-square model in Sec.~\ref{sec:coupled-ring}.}
\label{table.octagon-square}
\end{table}

\begin{table}
\begin{center}
\begin{tabular}{|c|c|}
\hline
$m$ & $\epsilon_0^{\alpha} (m)$ \\
\hline
\hline
0 & 0  \\
\hline
1 & -1.096997\\
\hline
2 & -2.013783\\
\hline
3 & -2.629382\\
\hline
4 & -3.086229\\
\hline
5 & -3.415430\\
\hline
6 & -3.609035\\
\hline
7 & -3.415430\\
\hline
8 & -3.086229\\
\hline
9 & -2.629382\\
\hline
10 & -2.013783\\
\hline
11 & -1.096997\\
\hline
12 & 0\\
\hline
\end{tabular}
\end{center}
\caption{The lowest energies of $\pi$-flux model on a $12$-site cross-shaped cluster, as shown in Fig.~\ref{fig.pi_flux} (a).
Here $m$ is the number of particles on the cluster. 
It shows that $\epsilon_0^{\alpha}(m=6)$ is the lowest ground-state energy. The results are used in the proof in Sec.~\ref{sec:2D-flux}.}
\label{table.pi_flux}
\end{table}

\begin{table}[h]
\begin{center}
\begin{tabular}{|c||c|c|c|c|c|c|c|c|c|}
\hline m & 1 & 2 & 3 & 4 & 5 & 6 & 7 & 8\\ \hline $\epsilon_0^{(4)}(m)$
& 0 & 0 & 0 & 0.372605 & 1.838145 & 4.323487 & 8 & 12\\ \hline
\end{tabular}
\end{center}
\caption{Ground-state energy $\epsilon_0^{(4)}(m)$ of the cluster Hamiltonian
$\mathcal H_n^{(4)}$ for delta-chain model, as shown in Fig.~\ref{fig.deltachain.p4}, with $m$ particles in a cluster.
It shows the ground-state energy of the $8$-site cluster is strictly positive when there are no less than four particles on this cluster. The results are used in Sec.~\ref{sec:flat-band}.} 
\label{table.delta-chain}
\end{table}

\begin{table}
\begin{center}
\begin{tabular}{|c|c|}
\hline $m$ & $\epsilon_0^{\textrm{cluster}}(m)$\\ \hline \hline 1 & 0\\
\hline 2 & 0\\ \hline 3 & 0\\ \hline 4 & 0.311475\\ \hline 5 &
0.937767\\ \hline 6 & 1.706509\\ \hline 7 & 3.365207\\ \hline 8 &
5.196963\\ \hline 9 & 7.456468\\ \hline 10 & 10.393543\\ \hline 11 &
14\\ \hline 12 & 18\\ \hline
\end{tabular}
\end{center}
\caption{The lowest energies of cluster Hamiltonian $\calH^{\textrm{cluster}}$
on $12$-site ``Start of David'' shape as shown in Fig.~\ref{fig.star.david}, in sectors with different numbers
of particles $m$. It shows the ground-state energy of one cluster is strictly positive when the number of particles on this cluster $m\ge4$. The results are used in Sec.~\ref{sec:flat-band}.}
\label{table.kagome}
\end{table}

\begin{table}
\begin{center}
\begin{tabular}{|c|c|c|}
\hline $m$ & $\epsilon_{\bigtriangleup}^{\rm F}(m)$ &
$\epsilon_{\bigtriangleup}^{\rm B}(m)$ \\ \hline \hline 1 & $-t+2t'$
& $-t-2t'$ \\ \hline 2 & $-2t+4t'$ & $-t-4t'$\\ \hline 3 & $6t'$ &
$-6t'$\\ \hline
\end{tabular}
\end{center}
\caption{The ground-state energies of fermions and hard-core bosons on a
thick-line up triangle  as shown in Fig.~\ref{fig.hexagon_triangle}(a), where $m$ is the number of particles on a
triangle. The results are used in the proof of $\pi$-flux hexagon-triangle model in Appendix~\ref{app:coupled_triangles}.}
\label{table:hexagon-triangle}
\end{table}

\section{coupled triangles}
\label{app:coupled_triangles}

\begin{figure}
\centering
\subfigure[]{
\includegraphics[width=2.0in]{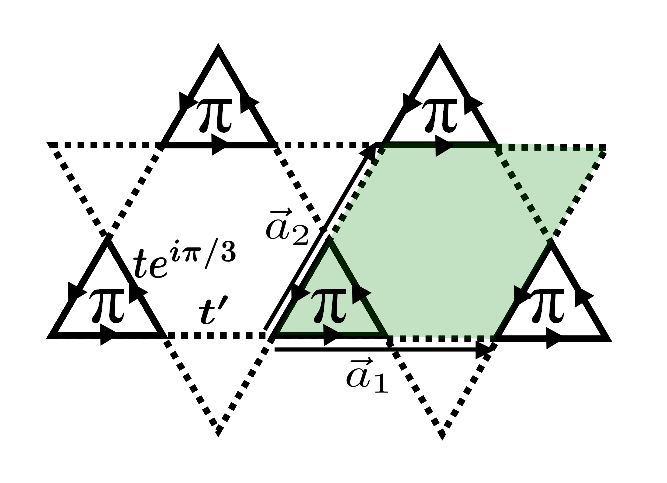}}
\subfigure[]{
\includegraphics[width=1.25in]{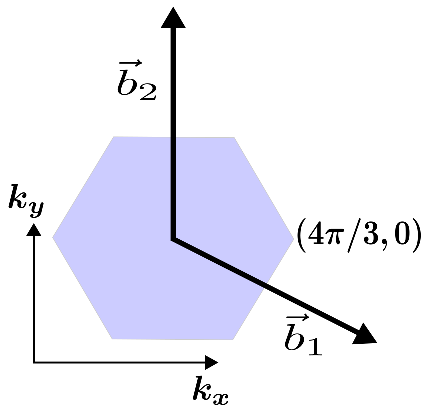}}\\
\subfigure[]{
\includegraphics[width=2.3in]{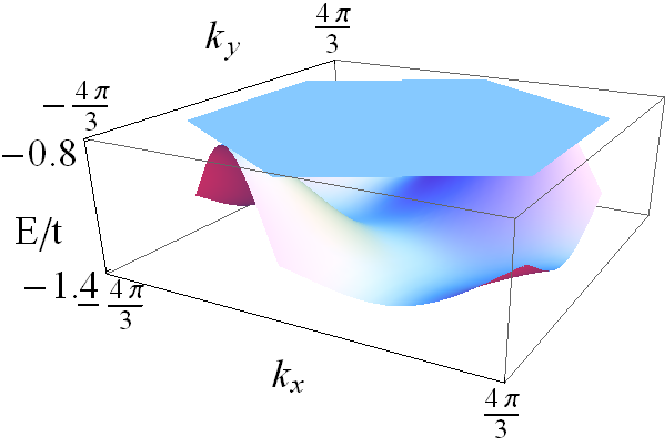}}
\caption{(a) $\pi$-flux hexagon-triangle lattice,
in which a unit cell is shown in green. (b) The first Brillouin zone.
The basis vectors are denoted by $\vec{b}_1$ and $\vec{b}_2$.
(c) Dispersions of the lowest two bands with $t=1$, $t'=0.2$.}
\label{fig.hexagon_triangle}
\end{figure}

The second example to show the natural inequality is reversed in coupled rings as in Sec.~\ref{sec:coupled-ring} is the $\pi$-flux hexagon-triangle
lattice, which is shown in Fig.~\ref{fig.hexagon_triangle} (a). This is actually a breathing kagome lattice. In the vanadium oxyfluoride compound $(\textrm{NH}_4)_2[\textrm{C}_7\textrm{H}_{14}\textrm{N}][\textrm{V}_7\textrm{O}_6\textrm{F}_{18}] (\textrm{DQVOF})$, the $\textrm{V}^{4+}$ ions realize a breathing kagome lattice~\cite{breathing_kagome}, topological equivalent to the hexagon-triangle as we discussed here. One unit
cell is shown in green in Fig.~\ref{fig.hexagon_triangle} (a), with basis
vectors $\vec{a}_1=(0,1)$ and $\vec{a}_2=(1/2,\sqrt{3}/2)$. The
Hamiltonian is defined as
\begin{equation}
\calH=-t \sum_{\langle i,j\rangle\in\textrm{thick, oriented}}
e^{i\pi/3}c_i^{\dag}c_j -t'\sum_{\langle
i,j\rangle\in\textrm{broken}}c_i^{\dag}c_j+\textrm{H.c.},
\end{equation}
where ``thick, oriented'' and ``broken'' links are
specified in Fig.~\ref{fig.hexagon_triangle}(a).
This model can be regarded as triangles with $\pi$-flux,
coupled by weak hopping $t'$.
To obtain a lower bound for the ground-state energy of bosons and
an upper bound for that of fermions,
the Hamiltonians are written as $\calH^{\rm F} = \sum_{\bigtriangleup}
h_{\bigtriangleup}^{\rm F}-A$ and $\calH^{\rm B} = \sum_{\bigtriangleup}
h_{\bigtriangleup}^{\rm B}+B$ with the same definitions of $A$ and $B$
in Eqs.~\eqref{eq.operatorA}\eqref{eq.operatorB}, where
$h_{\bigtriangleup}^{\rm F}=-t\sum_{i=1}^{3}
(e^{i\pi/3}c_i^{\dag}c_{i+1}+\textrm{H.c.})+2t'\sum_{i=1}^3
c_i^{\dag}c_i$ and $h_{\bigtriangleup}^{\rm B}=-t\sum_{i=1}^{3}
(e^{i\pi/3}c_i^{\dag}c_{i+1}+\textrm{H.c.})-2t'\sum_{i=1}^3
c_i^{\dag}c_i$, the cluster Hamiltonians defined on a solid-line
pointing up triangle. Therefore, we have $E_0^{\rm
B}\ge\sum_{\bigtriangleup_i}\epsilon_{\bigtriangleup_i}^{\rm B}$,
$E_0^{\rm F}\le\sum_{\bigtriangleup_i}\epsilon_{\bigtriangleup_i}^{\rm
F}$. The ground-state energies in given $m$-particle sectors are
displayed in Table~\ref{table:hexagon-triangle} in Appendix~\ref{app:tables}.
A lower bound for bosons is now given by
$E_0^{\rm B}\ge -N(t+4t')$ when $t'/t\le 1/2$ or $E_0^{\rm B}\ge -6Nt'$
when $1/2<t'/t<1$, where $N$ is the number of unit cells.
An upper bound for fermions is given by
$\tilde{E}_0^{\rm F}$, which also depends on the density pattern on the
whole lattice. At $2/3$ filling, we find $E_0^{\rm F}\le -2N(t-2t')N$.
According to the results of exact diagonalization on a cluster, we find
when $t'/t < 1/8$, $E_0^{\rm B}\ge -N(t+4t') > -2N(t-2t')\ge E_0^{\rm
F}$. Thus the reversal of the inequality is proved.

The second approach for the ground-state energy of fermions is
based on an exact evaluation.
The dispersion relations are ($t$=1 is assumed):
\begin{eqnarray}
E^{(1)}&=&\frac{1}{2}(1-t'-\sqrt{9(t')^2+6t'+9+8t'\Lambda (\vec
k)}),\nonumber\\ E^{(2)}&=&t'-1,\nonumber\\
E^{(3)}&=&\frac{1}{2}(1-t'+\sqrt{9(t')^2+6t'+9+8t'\Lambda (\vec
k)}),\nonumber
\end{eqnarray}
where $\Lambda (\vec k)=\cos{k_1}+\cos{k_2}-\cos{k_3}$,
$k_{1,2}=\vec{k}\cdot\vec{a}_{1,2} $ and $k_3=k_1-k_2$. The ground-state
energy of fermions at $2/3$ filling is given by the integral of the
lowest two bands in the Brillouin zone, which is shown in
Fig.~\ref{fig.hexagon_triangle} (c),
\begin{eqnarray}
E_{0}^{\rm F}&=&\sum_{k_x,k_y}\big[E^{(1)}(k_x,k_y)+E^{(2)}\big]\nonumber\\
&=&\frac{\sqrt{3}N}{2}\iint_{BZ}\frac{dk_x}{2\pi}\frac{dk_y}{2\pi}
\big[E^{(1)}(k_x,k_y)+E^{(2)}\big] ,
\end{eqnarray}
where $k_{x,y}\in BZ$ as shown in Fig.~\ref{fig.hexagon_triangle}
(b). The basis vectors $\vec{b}_1$ and $\vec{b}_2$ are chosen
accordingly as $2\pi(1,\;-1/\sqrt{3})$ and $2\pi(0,\;2/\sqrt{3})$, respectively.  The reversed natural inequality holds when $t'\ll t$. For example, when
$t=1$ and $t'=0.1$, $E_0^{\rm B}\ge -1.4N>E_0^{\rm F}=-2.004349N$; when
$t'=0.2$, $E_0^{\rm B}\ge -1.8N>E_0^{\rm F}=-2.017037N$.

\section{Optimal lower bound of filling fraction of the violation in delta-chain model}
\label{app:optimal_bound}

Let us improve the estimate of the range of the filling fraction,
for which the violation of Eq.~\eqref{eq.e0b.leq.e0f} occurs
on the delta-chain model, as discussed in Sec.~\ref{sec:flat-band}.
Our result is that the violation occurs, namely the reversed inequality
$E_0^{\rm B} > E_0^{\rm F}$ holds, for $1/4 < \nu \le 1/2$.  In fact, in this range
of filling, the ground-state energy of bosons is strictly positive while
the ground-state energy of fermions is zero.

To prove this, consider Bose-Hubbard model (without hard-core
constraint) with finite on-site $\calU>0$ in the enlarged Hilbert space
first,
\begin{eqnarray}
&&\calH = \calH_{\textrm{hop}}+\calH_{\textrm{int}}\nonumber,\\
&&\calH_{\textrm{hop}} = \sum_{j=1}^{N}a_j^{\dag}a_j\nonumber,\\
&&\calH_{\textrm{int}} = \frac{\cal U}{2}\sum_{i=1}^{2N}n_i (n_i-1)\nonumber,
\end{eqnarray}
where $n_i=c_i^{\dag}c_i$, and $[c_i,c_j^{\dag}]=\delta_{ij}$ for
bosons. The definition of $a$-operator is the same as
$a_j=c_{2j-1}+\sqrt{2}c_{2j}+c_{2j+1}$. The hard-core constraint can be
implemented by taking ${\cal U}\to\infty$, and this problem is reduced to
equation~\eqref{eq.deltachain} in this limit.

Obviously, the hopping term $\cal H_{\textrm{hop}}$ is positive
semi-definite. The on-site interaction, ${\cal U}$ term,
is also positive semi-definite
because $\frac{\cal U}{2} n_i (n_i-1)=\frac{\cal U}{2} c_i^{\dag}c_i^{\dag}c_i c_i$ for bosons. As a
consequence,
none of the eigenvalues can be negative.
Therefore, any state with $E^{\rm B}=0$ is a ground state.
If such a ground state $|\Phi_{\textrm{GS}} \rangle$
exists, it satisfies
\begin{equation}
\calH_{\textrm{hop}} | \Phi_{\textrm{GS}} \rangle =
\calH_{\textrm{int}} | \Phi_{\textrm{GS}} \rangle = 0,
\end{equation}
namely $|\Phi_{\textrm{GS}}\rangle$
a simultaneous zero-energy ground state of
$\calH_{\textrm{hop}}$ and $\calH_{\textrm{int}}$.
Therefore, we first seek zero-energy ground states
of $\calH_{\textrm{hop}}$ and $\calH_{\textrm{int}}$,
separately.

\begin{figure}
\begin{center}
\includegraphics[scale=0.6]{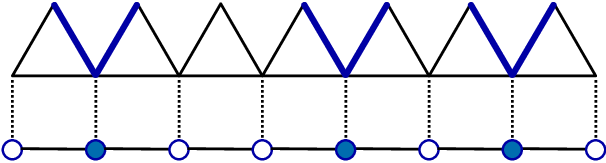} \caption{Schematic figure of
mapping to particle configurations in one-dimensional chain with nearest
neighbor exclusion. Localized zero-energy states(valley states)are shown
in blue lines.}  \label{fig.mapping1D}
\end{center}
\end{figure}

Localized-electron states are discussed in the field of flat-band
ferromagnet of Fermi-Hubbard
model~\cite{Mielke-kagome,Tasaki-flatband,Mielke-flatband,
DerzhkoPRB2007,*Richter2008-1D,*DerzhkoPRB2010}. We can construct the
localized state for bosons in a similar way.
Consider the zero-energy ground state of $\calH_{\textrm{hop}}$ first.
Define $b$-operator as $b_j=c_{2j}-\sqrt{2}c_{2j+1}+c_{2j+2}$. Because
$b$-operators commute with any $a$-operator, $[a_i,b_j^{\dag}]=0$ for
any $i$ and $j$, the single-particle flat band with $E_0^{\rm B}$ is
spanned by $b_j^{\dag}|0\rangle$. Note that these states
$b_j^{\dag}|0\rangle$ are linearly independent but not orthogonal to each other.
The zero-energy state (valley state) $b_j^{\dag}|0\rangle$ is shown in
Fig.~\ref{fig.mapping1D} by blue lines. It is the first excited state
of spin-$1/2$ antiferromagnetic Heisenberg model near saturation field,
with single magnon trapped in the valley of the
delta-chain~\cite{Schulenburg-PRL2002,Schnack-JPCM2004,Zhitomirsky-Honecker2004}.
The current setup corresponds to the magnetic field exactly
at the saturation field, so that these trapped magnons
are exactly at zero energy.
The ground state of $\calH_{\textrm{hop}}$ can be constructed out of
$b$-operators as,
\begin{equation}
|\Phi_0^{\rm B}\rangle = \sum_{\left\{n_1,\cdots,n_N\right\}}f(n_1,\cdots,n_N)
(b_1^{\dag})^{n_1}(b_2^{\dag})^{n_2}\cdots(b_N^{\dag})^{n_N}|0\rangle,
\label{eq.GS_Hhop}
\end{equation}
where $n_j=0,1,2,\cdots$ and $f(n_1,\cdots,n_N)$ is the coefficient. It
is easy to confirm $\calH_{\textrm{hop}}|\Phi_0^{\rm B}\rangle=0$,
by using the commutation relation $[a_i,b_j^{\dag}]=0$.

Now we require those zero-energy ground states~\eqref{eq.GS_Hhop} of
$\calH_{\textrm{hop}}$ to satisfy
$\calH_{\textrm{int}} | \Phi_0^{\rm B}\rangle = 0$.
This is equivalent to require $c_i c_i |\Phi_0^{\rm B}\rangle=0$,
which imposes restrictions on the coefficients
$f(n_1,\cdots,n_N)$.
We first note that {\setlength\arraycolsep{2pt}}
\begin{eqnarray}
c_{2j+1}^2|\Phi_0^{\rm B}\rangle & = &
\sum_{\left\{n_1,\cdots,n_N\right\}}2n_j(n_j-1)f(n_1,\cdots,n_N)\times\nonumber\\
&&(b_1^{\dag})^{n_1}\cdots(b_j^{\dag})^{n_j-2}\cdots(b_N^{\dag})^{n_N}|0\rangle.
\end{eqnarray}
Then the linear independence of $b$-operators, together with
$c_{2j+1}^2|\Phi_0^{\rm B}\rangle=0$, implies that
$f(n_1,\cdots,n_N)=0$ if there exists $j$ such that $n_j >1$.
We thus restrict our attention to the case where $n_j=0$ or $1$ for all $j$ in the sum. We successively find
{\setlength\arraycolsep{2pt}}
\begin{eqnarray}
c_{2j}^2|\Phi_0^{\rm B}\rangle & = &
\sum_{\left\{n_1,\cdots,n_N\right\}}2n_{j-1}n_j
f(n_1,\cdots,n_N)\times\nonumber\\
&&(b_1^{\dag})^{n_1}\cdots(b_{j-1}^{\dag})^{n_{j-1}-1}
(b_{j}^{\dag})^{n_{j}-1}\cdots(b_N^{\dag})^{n_N}|0\rangle,\nonumber\\
\end{eqnarray}
where $n_j=0$ or $1$ has been applied. From the linear independence of
$b$-operators and $c_{2j}^2|\Phi_0^{\rm B}\rangle=0$,
we see that $f(n_1,\cdots,n_N)=0$ if there exists $j$ such that $n_j n_{j-1}\ne0$.
This implies that, for bosons, in the construction of
the zero-energy ground state, $b^\dagger$-operators on adjacent
valleys cannot be applied on the vacuum $|0\rangle$.
Thus, the zero-energy ground states are in
one-to-one correspondence with particle configurations in
one-dimensional chain with nearest neighbor exclusion. This mapping is
schematically shown in Fig.~\ref{fig.mapping1D}. In the range $\nu\le
1/4$, we can find a particle configuration that satisfies the exclusion
rule. However, in the case $\nu>1/4$ we cannot find such configuration,
implying the absence of zero-energy state.

The zero-energy ground states remain as ground states for any ${\cal U}>0$,
and hence in the limit ${\cal U}\to\infty$. Since the on-site ${\cal U}$ term is
positive semi-definite, no state joins the zero-energy sector with
increasing ${\cal U}$. Therefore, the ground-state energy of hard-core bosons
(corresponding to infinite ${\cal U}$) is strictly positive in the range of
filling $\nu>1/4$.

\begin{figure}
\begin{center}
\includegraphics[scale=0.6]{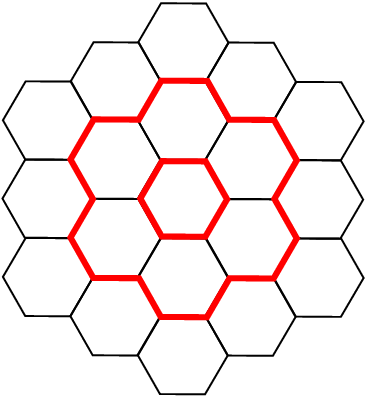} 
\caption{One example of uncontractible cycle sets on honeycomb lattice, which is constituted by two uncontractible cycles.}
\label{fig.cycle_sets}
\end{center}
\end{figure}

On the other hand, for fermions, $\{a_i,b_j^{\dag}\}=0$ holds for any
$i$ and $j$. The zero-energy state for fermions in the range of filling
fraction $\nu\le1/2$ can also be constructed by $b$ operators,
\begin{equation}
|\Phi_0^{\rm F}\rangle = \sum_{\left\{n_1,\cdots,n_N\right\}}f(n_1,\cdots,n_N)
(b_1^{\dag})^{n_1}(b_2^{\dag})^{n_2}\cdots(b_N^{\dag})^{n_N}|0\rangle,
\end{equation}
where $n_j=0$, $1$. It is easy to confirm that this is the zero-energy state
of $\calH$ because $\calH_{\textrm{hop}}|\Phi_0^{\rm F}\rangle=0$, and
$\calH_{\textrm{int}}$ vanishes.  We conclude the reversed inequality
$E_0^{\rm B} > E_0^{\rm F}$ holds in the range $1/4 < \nu \le 1/2$.

From the above analysis, it also follows that
both bosonic and fermionic systems have
exactly zero-energy ground states for $\nu \leq 1/4$.
Thus the lower bound of the range of the filling fraction
for the reversed inequality to hold, $1/4$, is in fact optimal.

An argument similar to the above for delta-chain model
can be employed for kagome lattice,
to extend the range of filling fraction where the natural inequality
is violated.
The zero-energy states for kagome lattice (the line graph of honeycomb lattice) are in one-to-one
correspondence with \textit{uncontractible cycle sets}
on the honeycomb lattice, as defined in Ref.~\onlinecite{Mielke-2012} in terms of graph theory.
An example of uncontractible cycle sets is shown in Fig.~\ref{fig.cycle_sets}.
It can then be deduced that the zero-energy states exist
for $\nu \leq 1/9$. The uncontractible cycle sets are given by
close-packed hard hexagons~\cite{Schulenburg-PRL2002,Zhitomirsky-Tsunetsugu,Mielke-2012}
at the critical value $\nu=1/9$,
and do not exist for $\nu> 1/9$.
Therefore, the range of filling fraction
in which $E_0^{\rm B}>E_0^{\rm F}$ holds on kagome lattice,
is extended to $1/9<\nu\le1/3$.

\bibliography{biblio}

\end{document}